\documentclass[acmsmall,screen,nonacm]{acmart}

\usepackage{xspace}
\usepackage[textsize=small]{todonotes}
\usepackage{bm}
\usepackage{listings}
\usepackage{subcaption}
\usepackage{xcolor}
\usepackage[inline]{enumitem}
\usepackage{cprotect}
\usepackage[capitalise,noabbrev]{cleveref}
\Crefname{figure}{Figure}{Figures}
\usepackage{rotating}
\usepackage{adjustbox}
\usepackage{multirow}

\AtBeginEnvironment{appendices}{\crefalias{section}{appendix}}
\usepackage{thm-restate}

\newcommand{\rw}[1]{\todo[inline,color=cyan!60]{\sf RW: #1}}
\newcommand{\nop}[1]{}

\newcommand{\cd}{\operatorname{:-}}
\newcommand{\set}[1]{\{#1\}}
\newcommand{\setof}[2]{\{#1 \mid #2\}}

\newcommand{\YA}{\textsf{YA}\xspace}
\newcommand{\HJ}{\textsf{HJ}\xspace}
\newcommand{\TTJ}{\textsf{TTJ}\xspace}
\newcommand{\TTJNG}{\textsf{TTJ}^{ng}\xspace}
\newcommand{\TTJNGDP}{\textsf{TTJ}^{ng+dp}\xspace}
\newcommand{\TTJDP}{\textsf{TTJ}^{dp}\xspace}
\newcommand{\TTJLinear}{\textsf{TTJ}^L\xspace}
\newcommand{\ttj}{\texttt{ttj}}

\newcommand{\IN}{\textsf{IN}}
\newcommand{\OUT}{\textsf{OUT}}

\newcommand{\keys}{\textsf{keys}}
\newcommand{\imdbcastinfo}{\mathsf{cast\_info}}

\newcommand{\tpchsupplier}{\mathsf{supplier}}
\newcommand{\tpchnation}{\mathsf{nation}}
\newcommand{\tpchlineitem}{\mathsf{lineitem}}
\newcommand{\tpchorders}{\mathsf{orders}}

\newcommand{\tout}{t_{\text{out}}}

\newcommand{\concat}{\mathbin{+\mkern-5mu+}}
\newcommand{\tconv}{\mathscr{T}}

\title{ TreeTracker Join: Simple, Optimal, Fast}

\author{Zeyuan Hu}
\affiliation{%
    \institution{University of Texas at Austin}
    \city{Austin}
    \state{Texas}
    \country{USA}
}

\author{Remy Wang}
\affiliation{%
    \institution{University of California, Los Angeles}
    \city{Los Angeles}
    \state{California}
    \country{USA}
}

\author{Daniel P. Miranker}
\affiliation{%
    \institution{University of Texas at Austin}
    \city{Austin}
    \state{Texas}
    \country{USA}
}

\begin{abstract}
We present a novel linear-time acyclic join algorithm,
TreeTracker Join (\TTJ). 
The algorithm can be understood as 
the pipelined binary hash join with a simple twist:
upon a hash lookup failure, \TTJ resets execution to the binding of the tuple 
causing the failure, and removes the offending tuple from its relation.
Compared to the best known linear-time acyclic join algorithm, 
Yannakakis's algorithm,
\TTJ shares the same asymptotic complexity
while imposing lower overhead. 
Further, we
prove that when measuring query performance by counting the number of hash
probes, \TTJ will match or outperform binary hash join on the same plan.
This property holds independently of the plan and independently of acyclicity. We
are able to extend our theoretical results to cyclic queries by introducing a new hypergraph
decomposition method called tree convolution.  Tree convolution iteratively identifies
and contracts acyclic subgraphs of the query hypergraph. The method avoids
redundant calculations associated with tree decomposition and may be of
independent interest. Empirical results on TPC-H, the Join Order Benchmark, and
the Star Schema Benchmark demonstrate favorable results.
\end{abstract}

\begin{document}

\maketitle

\section{Introduction}

Yannakakis~\cite{DBLP:conf/vldb/Yannakakis81} was the first to
describe a linear-time join algorithm (hereafter \YA) running in time $O(|\IN| + |\OUT|)$, 
where $|\IN|$ is the input size and $|\OUT|$ is the output size.
In principle, this is the best asymptotic complexity one can hope for,
because in most cases the algorithm must read the entire input and write the entire output.
However, virtually no modern database systems implement \YA.
A major factor is its high overhead. Prior to executing the join, \YA performs two passes over the input relations, using semijoins to reduce the input size. The reduction is lossless and enables optimally joining the reduced relations.
Since the cost of a semijoin is proportional to the size of its arguments, 
this immediately incurs a $2\times$ overhead in the input size.
An improved version of \YA~\cite{Bagan2007OnAC} achieves the same result in one semijoin pass,
but the overhead of this pass remains.
Another practical challenge is that \YA is ``too different'' from traditional binary join algorithms, 
making it difficult to integrate into existing systems.
For example, the efficiency of \YA critically depends on a query's {\em join tree}
which is different 
from the query plan used by binary joins%
\footnote{By {\em join tree} we mean (hyper-)tree decomposition of hypertree width 1,
not the tree of binary join operators commonly seen in relational algebra query plans.}. 
Where there is a wealth of techniques to optimize
query plans for binary joins, 
little is known about cost-based optimization of join trees for \YA.

\begin{figure*}[htbp]
    \centering
    \begin{minipage}[t]{0.35\textwidth}  
        \centering
        \begin{subfigure}[t]{\textwidth}
            \centering
\begin{lstlisting}[escapechar=|,showlines=true]
for i,x in R: 
  for y,j in S[x]:
    for k in T[y]: |\label{lst:k-loop}|
      # lines left blank
      # intentionally 
      for l in U[y]: |\label{lst:u-lookup}|
        print(x,y,i,j,k,l)
\end{lstlisting}
            \caption{Binary join}
            \label{fig:binary-exe}
        \end{subfigure}
    \end{minipage}%
    \hspace{-3.5em}
    \begin{minipage}[t]{0.3\textwidth}  
        \centering
        \begin{subfigure}[t]{\textwidth}
            \centering
\begin{lstlisting}[showlines=true,escapeinside={(*}{*)},numbers=none,basicstyle=\ttfamily\small\color{black!30}]
for i,x in R: 
  for y,j in S[x]:
    for k in T[y]:
      (*\color{black}\underline{\textbf{if} U[y] \textbf{is} None:}*)
        (*\color{black}\underline{\textbf{break}}*)
      for l in U[y]:
        print(x,y,i,j,k,l)
\end{lstlisting}
            \caption{Backjumping}
            \label{fig:backjump}
        \end{subfigure}
    \end{minipage}%
    \hspace{-1.5em}
    \begin{minipage}[t]{0.35\textwidth}  
        \centering
        \begin{subfigure}[t]{\textwidth}
            \centering
\begin{lstlisting}[showlines=true,escapeinside={(*}{*)},numbers=right,basicstyle=\ttfamily\small\color{black!30}]
for i,x in R: 
  for y,j in S[x]:
    for k in T[y]:
      if U[y] is None:
        (*\color{black}\underline{S[x].del((y,j));}*) break
      for l in U[y]:
        print(x,y,i,j,k,l)
\end{lstlisting}
            \caption{Tuple deletion}
            \label{fig:tuple-delete}
        \end{subfigure}
    \end{minipage}
\caption{Instantiation of binary hash join on
example~\ref{ex:main},
with backjumping,
and with tuple deletion.}
\Description{The figure shows the execution of a binary hash join algorithm, 
with backjumping, and with tuple deletion.}
\label{fig:main-idea}
    \label{fig:main}
\end{figure*}


In this paper, we propose a new linear-time join algorithm called
TreeTracker Join (\TTJ), inspired by the TreeTracker algorithm~\cite{DBLP:journals/ai/BayardoM94}
in Constraint Satisfaction. \TTJ can be understood as the traditional
binary hash join with a twist: when a hash lookup fails,
backtrack to the tuple causing the failure, and remove that tuple
from its relation.
The backtracking points depend only on the query, not the data, and are determined by the query compiler prior to query execution.
The execution deviates from binary hash join only when a dangling tuple is detected and deleted. Thus the tuple is excluded from any computation going forward. Hence, when using identical query plans \TTJ 
is guaranteed to match or outperform binary hash join (Section~\ref{sec:comparison}).


The following example illustrates the main ideas of \TTJ.
\begin{example}\label{ex:main}
Consider the natural join of the relations $R(i, x)$,
 $S(x, y, j)$, $T(y, k)$, and $U(y, l)$,
where we use $R(i, x)$ to denote that the schema of $R$ is $\{i, x\}$. The set $\set{1, \ldots, N}$ is denoted by $[N]$.
Let the relations be defined as follows:
$$
R = \setof{(i, 1)}{i \in [N]}\hspace{1em} S = \setof{(1, 1, j)}{j \in [N]}\hspace{1em} %
T = \setof{(1, k)}{k \in [N]}\hspace{1em} U = \setof{(0, l)}{l \in [N]}
$$
Observe $U$ 
shares no common $y$-values with $S$ or $T$, making the query result empty.
We'll first consider execution with binary hash join, the foundation of our algorithm.
Assume the optimizer produces a left-deep join plan $((R \bowtie S) \bowtie T) \bowtie U$. 
The execution engine builds hash tables for $S$, $T$, and $U$,
mapping each $x$ to $(y,j)$ values in $S$, $y$ to $k$ values in $T$, and $y$ to $l$ values in $U$.
Figure~\ref{fig:binary-exe}%
\footnote{One may also recognize this as indexed nested loop join, 
which is equivalent~\cite{sqliteoptoverview}.} illustrates the basis of the execution.
For each $(i, x)$ tuple in $R$, the hash table for $S$ is probed for the $(y,j)$ values.
T is probed with each pair $(y,j)$ to determine  the $k$ values. This repeats for  $U$.
Although the query produces no output, the execution takes $\Omega(N^3)$ time 
because it first computes the join of $R$, $S$, and $T$.
A closer look at the execution reveals the culprit:
when the lookup on $U$ produces no result, (line~\ref{lst:u-lookup}),
the algorithm continues to the next iteration of the loop over $k$ values, (line~\ref{lst:k-loop}). The same value of $y$ is used to probe into $U$ again!
To address this, \textbf{the first key idea of \TTJ is to backjump%
\footnote{Backjumping is a concept in backtracking search algorithms; 
we use the term informally to mean the interruption of 
a nested loop iteration to jump back to an outer loop, 
while referring to the original TreeTracker algorithm~\cite{DBLP:journals/ai/BayardoM94}
for a precise definition.}
to the level causing
the probe failure.} 
For clarity
we use \lstinline|break| to represent the backjump,
as shown in Figure~\ref{fig:backjump}.
When probing $U$ with the key value $y $ fails to return a result, 
we break out of the current loop over $k$
and continue to the next iteration
of the second loop level. That is because the unsuccessful lookup key value of $y$ is assigned at that level. That next iteration retrieves 
new $y,j$ values,
skipping over iterations for $k$ values
that are doomed to fail.
With this optimization, the execution finishes in $O(N^2)$ time,
as it still needs to compute the join of $R$ and $S$.
To improve the performance further:
\textbf{the second key idea of \TTJ is to delete the tuple causing the probe failure.}
This is shown in Figure~\ref{fig:tuple-delete}:
after the probe failure, the offending tuple $(x, y, j)$ is removed from the $S$ hash table.
This is safe to do, because that $y$ value will always fail to join with any tuple in $U$.
In this way, all tuples from $S$ are removed after looping over it the first time.
Then, on all subsequent iterations of the loop over $R$, the probe into $S$ 
fails immediately.
Overall, the algorithm finishes in $O(N)$ time.
\end{example}
%
In general, \TTJ runs in linear time in the size of the input and output
for full acyclic joins. But the algorithm is not limited to acyclic queries:
given the same query plan for {\em any} query, cyclic or acyclic, \TTJ is guaranteed to match the performance of binary hash join,
when measuring query performance by counting the number of hash probes.
In particular, when no probe fails \TTJ behaves identically to binary join.
This is in contrast to \YA which always carries the overhead of semijoin reduction,
even if the reduction does not remove any tuple.

To address cyclic queries further we introduce a new method to break down a cyclic query into acyclic parts called {\em tree convolution}.  
We use this method to analyze the run time of \TTJ. 
A special kind of tree convolution, called {\em rooted convolution},
eliminates materialization of intermediates during query processing.

In summary, our contributions include:
\begin{itemize}
    \item Propose \TTJ, a new join algorithm that runs in time $O(|\IN| + |\OUT|)$ on full acyclic queries.
    \item Prove that \TTJ matches or outperforms binary join given the same query plan, on both acyclic and cyclic queries.
    \item Introduce {\em tree convolution}, a new method to break down cyclic queries into acyclic parts,
    and use it to analyze the run time of \TTJ on cyclic queries.
    \item Improve the performance of \TTJ with further optimizations.
    \item Conduct experiments to evaluate the efficiency of \TTJ on acyclic queries.
\end{itemize}

\section{Related Work}
The observation that only one semijoin pass is necessary in \YA 
has been a folklore in the database community,
with an early appearance in the theoretical work of 
Bagan, Durand, and Grandjean~\cite{DBLP:conf/csl/BaganDG07}.
Their paper studies the problem of enumerating conjunctive query results
with constant delay, but without considering practical efficiency.
Recent systems implementing such enumeration algorithms
take advantage of the same insight~\cite{DBLP:journals/corr/abs-2205-05649, DBLP:journals/sigmod/TziavelisGR24}.
Compared to this improved version of \YA, \TTJ has the guarantee of
matching the performance of \HJ given any query plan,
and is often faster in practice, as we will show in Section~\ref{sec:experiments}.

Researchers have also explored ways to integrate elements 
of \YA into existing systems.
Zhu et.al.~\cite{DBLP:journals/pvldb/ZhuPSP17}
propose {\em lookahead information passing},
using bloom filters to implement semijoins over star schemas.
Birler, Kemper, and Neumann~\cite{DBLP:journals/pvldb/BirlerKN24}
decompose every join operator into a {\em lookup} and 
an {\em expand}, and prove that certain lookup-and-expand (L\&E) plans
are guaranteed to run in linear time for acyclic queries.
Bekkers et.al.~\cite{DBLP:journals/corr/abs-2411-04042}
implement L\&E plans in a vectorized query engine,
while proving that their approach is guaranteed to match 
binary hash join for a class of {\em well-behaved} query plans.
The theoretical guarantees of \TTJ is complementary to these approaches:
while \TTJ guarantees to match binary hash join for left-deep plans,
the well-behaved class defined by Bekkers et.al.~\cite{DBLP:journals/corr/abs-2411-04042} essentially contains right-deep plans with
a slight generalization.
On the other hand, as we focus on an algorithm-level evaluation of \TTJ in this paper,
our implementation is not yet competitive with the highly optimized systems mentioned above.
Future work shall explore how to incorporate various system-level optimizations
like query compilation, vectorization, and parallelization
into \TTJ.

Going beyond acyclic queries, the standard way to handle cyclic queries
is to break up the query with (hyper-)tree decomposition~\cite{Gottlob2016}.
Such decomposition results in smaller cyclic subqueries connected by an acyclic ``skeleton''.
Each cyclic subquery can then be computed with worst-case optimal join 
algorithms~\cite{ngo2018wcoja,DBLP:journals/corr/abs-1210-0481}.
With the result of each subquery materialized, the final output can then be
computed with \YA.
As we will show in \cref{sec:cyclic}, \TTJ can support cyclic queries with only
a few modifications.
Compared to the tree decomposition approach, \TTJ does not require materializing
intermediate results, thus requiring only constant space in addition to the linear
space required to store and index the input relations.
While the worst-case time complexity of \TTJ does not match that obtained by
tree decompositions, the advantage of each approach depends on the data.

As the name suggests, TreeTracker Join is a direct decendent of the TreeTracker
algorithm~\cite{DBLP:journals/ai/BayardoM94} from Constraint Satisfaction. The
TreeTracker CSP algorithm resolved Dechter's
conjecture~\cite{DBLP:journals/ai/Dechter90} that there existed an optimal
algorithm for acyclic CSPs free of any preprocessing. The connection between
query answering and constraint satisfaction is a recurring theme in the literature to the extent that an expression emerged, the problems are two sides of the same
coin~\cite{DBLP:journals/jcss/KolaitisV00,Miranker1997QueryEA}. There are substantive differences that make TreeTracker and \TTJ different. First the constraint satisfaction problem concerns the existence of a non-empty model for a large logical formula. Thus, constraint satisfaction algorithms including TreeTracker stop execution and return \textsf{TRUE} upon identifying what in a relational query would be just one row of the result. In contrast, \TTJ produces all tuples in the query output. Second, the TreeTracker algorithm does not make use of hash tables, and is thus structured like a nested loop join rather than a hash join. 
This is because unlike the
study of queries in databases, constraint satisfaction rarely specializes the
problem to only equality predicates. Combining these two differences TreeTracker incorporates ad-hoc data structures, where \TTJ employs recognized indices commonly used in databases. These difference clearly manifest in the respective complexity analyses. 
The complexity of the best variation of the TreeTracker algorithm is polynomial in the input size and does not
consider the output size.  We prove below \TTJ runs in linear time in the total
size of the input and output.
\section{Preliminaries}\label{sec:preliminaries}
In this section, we present the foundational concepts concerning acyclic join queries and the specific definitions adopted in this paper.

\subsection{Join Queries and Acyclicity}
We consider natural join queries, also known as  \emph{full conjunctive queries}, of the form:
\begin{equation}\label{eq:query}
Q(\bm{x}) = R_1(\bm{x}_1) \bowtie R_2(\bm{x}_2) \bowtie \cdots \bowtie R_n(\bm{x}_n)
\end{equation}
where each $R_i$ is a relation name, each $\bm{x}_i$ (and $\bm{x}$) a tuple of distinct variables,
and every $x \in \bm{x}_i$ also appears in $\bm{x}$.
We call each $R_i(\bm{x}_i)$ an {\em atom},
and $x_i$ the {\em schema} of $R_i$,
denoted as $\Sigma(R_i)$.
We extend the notion of schema to tuples in the standard way
and write $\Sigma(t)$ for the schema of $t$.
The query computes the set%
\footnote{For clarity we assume set semantics. No change is needed for \TTJ to support bag semantics}
$Q = \setof{\bm{x}}{\bigwedge_{i \in [n]} \bm{x}_i \in R_i}$.
We sometimes write $Q$ and not $Q(\bm{x})$ to reduce clutter,
and identify $Q$ with its set of relations.
For example, $Q - \{R_i\}$ denotes the query $Q$ with $R_i$ removed.

%
%
%
\begin{definition}[Join Tree]\label{def:join-tree}
A \emph{join tree} for a query $Q$ is a tree where each node is an atom in $Q$,
such that for every variable $x$, the nodes containing $x$ form a connected subtree.
\end{definition}
%
A query $Q$ is \emph{acyclic} (more specifically $\alpha$-acyclic) if there 
exists a {\em join tree} for Q.

For clarity we rewrite the query in example~\ref{ex:main} and detail one of its join trees:
\begin{equation}\label{eq:main}
Q_1 = R(i, x) \bowtie S(x, y, j) \bowtie T(y, k) \bowtie U(y, l)
\end{equation}
One join tree has $R(i,x)$ at the root, $S(x,y,j)$ as its child, 
and $T(y,k)$ and $U(y,l)$ as children of $S$.
We encourage the reader to draw a picture of this join tree for reference.
One can construct a join tree for any acyclic query with the GYO 
algorithm~\cite{Yu1979AnAF,graham1980universal},
which works by finding a sequence of {\em ears}.
To define ear, we first define a {\em key schema}:
\begin{definition}[Key Schema]
For a query $Q$ of the form~\eqref{eq:query},
the {\em key schema} of an atom $R_i(\bm{x}_i)$ in $Q$,
denoted as $\keys(Q, R_i)$, is the set of variables
shared between $R_i(\bm{x}_i)$ with the other atoms in $Q$;
i.e., $\keys(Q, R_i) = \bm{x}_i \cap \bigcup_{j \in [n] \wedge j\neq i} \bm{x}_j$.
\end{definition}
Intuitively, $\keys(Q, R_i)$ form the keys of $R_i$'s hash table, 
if we compute $(Q - \set{R_i}) \bowtie R_i$ using binary hash join.
\begin{definition}[Ear]\label{def:ear}
Given a query $Q$ of the form~\eqref{eq:query},
an atom $R_i(\bm{x}_i)$ is an {\em ear} if it
satisfies the property $\exists j\neq i: \bm{x}_j \supseteq \keys(Q, R_i)$.
In words, there is another atom $R_j(\bm{x}_j)$ that contains all
the variables in $R_i$'s key schema.
We call such an $R_j$ a \emph{parent} of $R_i$. 
\end{definition}
The parent concept is central to the \TTJ algorithm. The parent's schema include all
of its children's keys. When a hash lookup fails at a child,
\TTJ will backjump to the parent.
Figure~\ref{fig:parent} shows an algorithm to find
the first parent of an ear in $Q$,
where $Q$ is represented as a list of atoms.


The GYO algorithm for constructing join trees is shown in Figure~\ref{fig:gyo}:
we start with a forest where each atom makes up its own tree,
then for every ear, we attach it to its parent and remove that ear from the query.
Note that it is possible for the algorithm to produce a forest of
disjoint trees when the query contains Cartesian products.
For simplicity, we will ignore such cases.
\begin{figure}
\begin{subfigure}[t]{0.45\textwidth}
\begin{lstlisting}
def GYO(Q):
  forest = { tree(R) for R in Q }
  while not Q.is_empty():
    R = find-ear(Q)
    P = parent(Q, R)
    forest.set_parent(R, P)
    Q.remove(R)
  return forest
\end{lstlisting}
\caption{The GYO alglorithm.}
\label{fig:gyo}
\end{subfigure}
\begin{subfigure}[t]{0.5\textwidth}
\begin{lstlisting}[showlines=true]
def parent(R, Q):
  if Q.is_empty(): return None
  keys = $\Sigma\text{(R)} \cap \bigcup_{\text{S} \in \text{Q-\{R\}}} \Sigma\text{(S)}$
  for S in Q - {R}:
    if keys $\subseteq$ $\Sigma\text{(S)}$: 
      return S
  # we did not find a valid parent
  return None
\end{lstlisting}
\caption{Find a parent of $R$ in $Q$ if one exists.}
\label{fig:parent}
\end{subfigure}
\caption{GYO reduction and parent computation.}
\end{figure}
\begin{definition}[GYO reduction order]\label{def:gyo}
Given a query $Q$ of the form~\eqref{eq:query},
a {\em GYO reduction order} for a query $Q$
is a sequence $[R_{p_1}, R_{p_2}, \ldots, R_{p_n}]$ that is a permutation 
of $[R_1, R_2, \ldots, R_n]$, such that for every $i<n$, 
the atom $R_{p_i}$ is an ear in the (sub)query $R_{p_i} \bowtie \cdots \bowtie R_{p_n}$.
\end{definition}
Equivalently, it is the same order of atoms as visited by the GYO algorithm.
The reader can verify $[U,T,S,R]$is a GYO reduction order for $Q_1$.
The existence of a GYO reduction order and the existence of a join tree are equivalent.

%
\begin{theorem}[\cite{Yu1979AnAF,graham1980universal}]\label{thm:gyo}
A query $Q$ has a join tree (i.e., $Q$ is $\alpha$-acyclic) if and only if it has a GYO reduction order.
\end{theorem}
%
%
\subsection{Binary Hash Join}

In this paper we focus on hash-based join algorithms.
For theoretical analyses we focus on left-deep linear plans;
for practical implementation we follow the standard practice 
and decompose each bushy plan into a sequence of left-deep linear plans,
materializing each intermediate result.
\begin{definition}[Query Plan]
A (left-deep linear) query plan for a query $Q$ of the form~\eqref{eq:query}
is a sequence $[R_{p_1}, R_{p_2}, \ldots, R_{p_n}]$ 
that is a permutation of $Q$'s relations $[R_1, R_2, \ldots, R_n]$.
\end{definition}
For consistency we adopt 1-based indexing for query plans,
so the first relation in the plan is stored at $i=1$.
An example query plan for $Q_1$ in~\ref{eq:main} is $[R,S,T,U]$.
One may notice similarities between a GYO reduction order and a query plan.
The reason for this will become clear.

\begin{figure*}
\begin{subfigure}[t]{0.50\textwidth}
\centering
\begin{lstlisting}[escapechar=|,showlines=true]
def join(t, plan, i):
  if i > plan.len(): print(t)
  else: 
    R = plan[i]; k = $\pi_{\keys(\text{plan[1..i],R})}(t)$
    for r in R[k]:
      join(t$\concat$r, plan, i+1)
\end{lstlisting}
\caption{Pipelined left-deep binary hash join}
\label{fig:binary-join}
\end{subfigure}%
\begin{subfigure}[t]{0.50\textwidth}
\begin{lstlisting}
def YA(Q, order):
  for R in order: # semijoins reduction
    P = parent(Q, R); Q.remove(R)
    if P is not None: P = P $\ltimes$ R
  # compute the output with hash join
  return join((), reverse(order), 1)
\end{lstlisting}
\caption{Yannakakis's algorithm}
\label{fig:ya}
\end{subfigure}
\caption{Binary hash join and Yannakakis's algorithm. The \lstinline|plan| array is 1-indexed.}
\end{figure*}

We follow the push-based model~\cite{DBLP:journals/pvldb/Neumann11} and specialize the 
binary hash join algorithm for pipelined left-deep plans as shown in Figure~\ref{fig:binary-join}.
We write $\pi_s(t)$ to project the tuple $t$ onto the schema $s$,
and $t\concat r$ to concatenate the tuples $t$ and $r$ while resolving
the schema appropriately.
Execution begins by passing to \lstinline|join| the empty tuple $t = ()$,
a query plan, and $i = 1$.
Although we do not need to build a hash table for the left-most relation
(the first relation in the plan), 
for simplicity we assume that there is a (degenerate) hash table
mapping the empty tuple $()$ to the entire left-most relation.
The algorithm starts by checking if the plan has been exhausted and if so, 
output the tuple $t$.
Otherwise, we retrieve the $i$-th relation $R_{p_i}$ from the plan,
and lookup from $R_{p_i}$ the matching tuples that join with $t$.
For each match, we concatenate it with $t$ and recursively call \lstinline|join|.


It may be helpful to unroll the recursion over a query plan, 
and we encourage the reader to do so for $Q_1$ in~\eqref{eq:main}
with the plan $[R,S,T,U]$. This will generate the 
same code as in Figure~\ref{fig:binary-exe}.


\subsection{Yannakakis's Algorithm}
Yannakakis's original algorithm~\cite{DBLP:conf/vldb/Yannakakis81}
makes two preprocessing passes over the input relations. A third pass 
computes the joins yielding  the final output.
Bagan, Durand, and Gandjean~\cite{Bagan2007OnAC} improved the original algorithm by
eliminating  the second preprocessing pass.
For brevity we only describe the latter algorithm.  Following common usage, hereafter, we will refer to the improved version as Yannakakis's alglorithm (\YA).

Shown in Figure~\ref{fig:ya} is, given a GYO reduction order,
the relations are preprocessed using semijoins, 
then the output is computed with standard hash join.
Equivalently, the semijoin preprocessing step can be performed by traversing a join tree bottom-up,
and the output computed  with hash join by traversing the tree top-down.
\begin{example}
Given the query $Q_1$ in~\eqref{eq:main} and the GYO reduction order $[U,T,S,R]$,
\YA first performs the series of semijoins, $S'= S\ltimes U$, $S'' =S'\ltimes T$, 
and $R'= R\ltimes S''$, then computes the output with the plan $[R', S'', T, U]$.
The reader may refer to the join tree of $Q_1$ and confirm we are 
traversing the tree bottom-up then top-down.
\end{example}

\section{TreeTracker Join}\label{sec:algorithm}

The TreeTracker Join algorithm is shown in Figure~\ref{fig:ttj}.
The algorithm follows the same structure as binary hash join.
The difference starts on line~\ref{lst:ttj-parent}
right before the hash lookup \lstinline|R[k]|.
If this lookup fails (i.e., it finds no match),
and if $R$ has a parent $P$ that appears before $R$ in the plan,
then \TTJ backjumps to the \lstinline|for|-loop at $P$'s recursive level,
by returning \lstinline|P| (line~\ref{lst:backjump}).
This is similar to throwing an exception
which is ``caught'' at the loop level of \lstinline|P|,
as we will explain on line~\ref{lst:ttj-catch}.
Otherwise, if the lookup \lstinline|R[k]| succeeds,
the algorithm iterates over each matching tuple $r$
and calls itself recursively (line~\ref{lst:ttj-recursive}).
This recursive call has three possible results.
A result containing a relation (line~\ref{lst:ttj-catch})
signifies a backjump has occurred,
with that relation as the backjumping point.
If the backjumping point is the same as the current
relation $R$, then the tuple $r$ is deleted from $R$ (line~\ref{lst:ttj-delete}).
If the backjumping point is different from $R$,
then the backjump continues by returning \lstinline|result|
which interrupts the current loop.
Finally, if the recursive call (implicitly) returns \lstinline|None|,
the algorithm continues to the next loop iteration.

\begin{figure}
\begin{subfigure}[t]{0.5\textwidth}
\begin{lstlisting}[escapeinside={(*}{*)}]
def ttj(t, plan, i):
  if i > plan.len(): print(t)
  else: 
    R = plan[i]; k = $\pi_{\keys(\text{plan[1..i],R})}(t)$
    P = parent(plan[1..i], R)(*\label{lst:ttj-parent}*)
    if R[k] is None & P is not None:(*\label{lst:ttj-lookup}*)
      return P # backjump to P (*\label{lst:backjump}*)
    for r in R[k]: (*\label{lst:ttj-loop}*)
      result = ttj(t$\concat$r, plan, i+1)(*\label{lst:ttj-recursive}*)
      if result == R: # catch backjump (*\label{lst:ttj-catch}*)
        R[k].delete(r)(*\label{lst:ttj-delete}*)
      elif: result is not None:
        return result # continue backjump
\end{lstlisting}
\caption{TreeTracker join}
\label{fig:ttj}
\end{subfigure}%
\begin{subfigure}[t]{0.5\textwidth}
\begin{lstlisting}[basicstyle=\color{lightgray}\ttfamily\small,escapeinside={(*}{*)}, numbers=right]
if R[()] is None: throw BJ(None) (*\label{lst:throw1}*)
(*\color{black}{\textbf{for} i,x \textbf{in} R:}*)
  try: if S[x] is None: throw BJ(R) (*\label{lst:throw2}*)
    (*\color{black}{\textbf{for} y,j \textbf{in} S[x]:}*)
      (*\color{black}{\textbf{try:}}*) if T[y] is None: throw BJ(S) (*\label{lst:throw3}*)
        (*\color{black}{\textbf{for} k \textbf{in} T[y]:}*)
          try: (*\color{black}{\textbf{if} U[y] \textbf{is} None: \textbf{throw} BJ(S)*)
            (*\color{black}{\textbf{for} l \textbf{in} U[y]:}*)
              try: (*\color{black}{output(x,y,i,j,k,l)}*)
              catch BJ(U): U[y].delete(l)
          catch BJ(T): T[y].delete(k)
      (*\color{black}{\textbf{catch} BJ(S): S[x].delete(y,j)}\label{lst:catch-s}*)
  catch BJ(R): R[()].delete(i, x) (*\label{lst:catch1}*)
\end{lstlisting}
\caption{Execution of \TTJ for Example~\ref{ex:main}}
\label{fig:ttj-unrolled}
\end{subfigure}
\caption{The TreeTracker algorithm and an example execution.}
\end{figure}

\begin{example}\label{ex:ttj-unroll}
It can be helpful to unroll the recursive algorithm over a query plan.
Given $Q_1$ in~\eqref{eq:main} and the plan $[R,S,T,U]$,
Figure~\ref{fig:ttj-unrolled} shows the execution of \TTJ.
To make the code more intuitive, we replace \lstinline|return|
statements with exception handling to simulate backjumping.
We gray out dead code and no-ops:
\begin{itemize}
  \item Line~\ref{lst:throw1} is unreachable because \lstinline|R[()]| is always
the entire relation $R$, and $R$ has no parent.
  \item Line~\ref{lst:throw2} (and~\ref{lst:catch1}) is a no-op, because it would just backjump to 
the immediately enclosing loop, and removing a tuple from $R$ is useless 
because $R$ is at the outermost loop%
\footnote{In Section~\ref{sec:optimizations} we will introduce an additional 
optimization that makes ``removing'' from the outermost relation meaningful.}.
  \item Technically the if-statement on line~\ref{lst:throw3} is useful even though it 
only backjumps one level, because the backjump would remove a tuple from 
$S$ when caught (line~\ref{lst:catch-s}). However for the input data in Example~\ref{ex:main}
we do not need this, and we gray it out to reduce clutter.
  \item Finally, the innermost two try-catch pairs are unreachable,
because $U$ and $T$ have no children.
\end{itemize}
At this point, the remaining code in black is essentially the same
as the code in Figure~\ref{fig:tuple-delete}.
As a side note, a sufficiently smart compiler with partial evaluation
or just-in-time compilation
could remove the dead code and no-ops as we have done above.
\end{example}

\subsection{Correctness and Asymptotic Complexity}
The correctness proof starts with an observation on the
relationship between different calls to \lstinline|ttj|:
\begin{proposition}{\label{prop:ttj-calls}}
If $\ttj(t_j, p, j)$ recursively calls $\ttj(t_i, p, i)$, then $t_j \subseteq t_i$.
\end{proposition}
\begin{proof}
The proposition follows from the definition of the algorithm, where the $t_i$ argument
to the nested call is constructed by appending tuples to $t_j$.
\end{proof}
\TTJ differs from binary join only upon a lookup failure. In that case
it backjumps to the parent of the relation that caused the failure,
and deletes the tuple that caused the failure.
Therefore, \TTJ is correct as long as it never deletes or ``backjumps over''
any tuple that should be in the output.
We first prove that a deleted tuple
can never contribute to any outupt.
In the following we write $\pi_R(t)$ for the projection of $t$ 
onto the schema of $R$.
\begin{lemma}\label{lem:delete}
Suppose a tuple $r_j$ is deleted from $R_j$ during the execution of \TTJ for a query $Q$
using plan $p$.
Then $\forall \tout \in Q : \pi_{R_j}(\tout) \neq r_j$.
\end{lemma}
\begin{proof}
Let $p$ be $[R_1, \ldots, R_n]$,
and $t_j$ be the value of the argument $t$ in scope at the time of the deletion.
Because $r_j$ is deleted from $R_j$,  
there must be a failed lookup $R_i[k_i]$
recursively nested within the call to $\ttj(t_j \concat r_j, p, j+1)$,
and $R_j$ is the parent of $R_i$.
Let $K_i = \keys(p[1,\ldots, i],R_i)$,
and let $t_i$ be the value of $t$ at the time of the lookup failure.
Then $t_j \concat r_j \subseteq t_i$ by Proposition~\ref{prop:ttj-calls}.
By definition of parent, $K_i \subseteq \Sigma(R_j) \subseteq \Sigma(t_j \concat r_j) \subseteq \Sigma(t_i)$,  
so $k_i = \pi_{K_i}(t_i) = \pi_{K_i}(t_j \concat r_j) = \pi_{K_i}(r_j)$.
However, since the lookup failure implies no tuple in $R_i$ contains $k_i$, 
any output tuple $\tout$ cannot contain $k_i$ either, 
i.e., $\forall \tout \in Q : \pi_{K_i}(\tout) \neq k_i$.
Therefore, $\forall \tout \in Q : \pi_{K_i}(\tout) \neq \pi_{K_i}(r_j)$
which implies $\forall \tout \in Q : \pi_{R_j}(\tout) \neq \pi_{R_j}(r_j)$.
\end{proof}

Next, we show \TTJ never backjumps over any tuple that contributes to the output.
Given a plan $p = [R_1, \ldots, R_n]$,
denote by $\pi_{[i]}(t)$ the projection of $t$ onto $\bigcup_{j\in [i]} \Sigma(R_j)$.
%
\begin{lemma}\label{lem:ttj}
For any tuple $\tout \in Q$, plan $p$ for $Q$, and $1 \leq i \leq |p|$,
 $\ttj(\pi_{[i-1]}(\tout), p, i)$
recursively calls $\ttj(\pi_{[i]}(\tout), p, i+1)$.
\end{lemma}
\begin{proof}
Consider a lookup $R[k]$ that is recursively nested within the 
call to $\ttj(\pi_{[i-1]}(\tout), p, i)$
where $R$ has a parent $R_j$ with $j \in [i-1]$.
Then $k \subseteq \pi_{R_j}(\tout) \subseteq \tout$,
and because $\Sigma(k) \subseteq \Sigma(R)$,
we have $k\subseteq \pi_R(\tout) \in R$.
This means the lookup $R[k]$ will not fail.
This holds for all such $R$, so the algorithm never backjumps
from within the call $\ttj(\pi_{[i-1]}(\tout), p, i)$
to any $R_j$ for $j \in [i-1]$.
The algorithm may still backjump to $R_{i}$, but by Lemma~\ref{lem:delete},
$\pi_{R_{i}}(\tout)$ is never deleted from $R_{i}$, and therefore
the algorithm will recursively call $\ttj(\pi_{[i-1]}(\tout)\concat\pi_{R_i}(\tout), p, i+1)$
which is the same as $\ttj(\pi_{[i]}(\tout), p, i+1)$.
\end{proof}
%
%
We arrive at the correctness of \TTJ by applying Lemma~\ref{lem:ttj} inductively
over the query plan.
\begin{theorem}
Given any plan \lstinline|p| for $Q$, 
\lstinline|ttj((),p,1)| computes $Q$.
\end{theorem}
\begin{proof}
We prove the correctness of \TTJ in two directions:
first, any tuple produced by \TTJ should be in the output;
second, \TTJ produces all tuples that should be in the output.
The first direction is straightforward, as any tuple produced by \TTJ
is also produced by binary hash join.
We prove the second direction by induction over the argument $i$,
with the following inductive hypothesis:
$\ttj(\pi_{[i-1]}(\tout), p, i)$ will be invoked 
for all $\tout \in Q$ and $1\leq i\leq |p|$. 
The base case when $i = 1$ holds because we start the execution of \TTJ
by calling $\ttj((), p, 1)$.
For the inductive step, assume $\ttj(\pi_{[i-1]}(\tout), p, i)$
is invoked, then applying Lemma~\ref{lem:ttj} shows
$\ttj(\pi_{[i]}(\tout), p, i+1)$ will also be invoked.
Therefore, $\ttj(\tout, p, |p|+1)$ will be invoked for all $\tout \in Q$,
which produces all tuples that should be in the output.
\end{proof}

Next, we prove \TTJ runs in linear time in the size of 
the input and output, for full acyclic queries.
We first introduce a condition on the query plan
that is necessary for the linear time complexity:
%


\begin{lemma}\label{lem:gyo}
Given a query $Q$ and a plan $p = [R_{1}, \ldots, R_{n}]$ for $Q$,
\lstinline|parent| returns \lstinline|None| only for $R_{1}$
during the execution of \TTJ,
if $p$ is the reverse of a GYO reduction order of $Q$.
\end{lemma}
\begin{proof}
 If $[R_{n}, \ldots, R_{1}]$ is a GYO reduction order,
 then there is a join tree with $R_{1}$ as root, 
 and every non-root atom has a parent. 
\end{proof}
\TTJ is guaranteed to run in linear time given such a plan:
\begin{theorem}\label{thm:linear}
Fix a query $Q$ and a plan $p$. If $p$ is the reverse of 
a GYO reduction order for $Q$, then \lstinline|ttj((),p,1)|
computes $Q$ in time $O(|Q| + \sum_i |R_i|)$.
\end{theorem}
\begin{proof}
We first note that in Figure~\ref{fig:ttj},
\lstinline|ttj| does constant work outside of the loops; 
each iteration of the loop also does constant work 
and recursively calls \lstinline|ttj|,
so each call to \lstinline|ttj| accounts for constant work,
therefore the total run time is linear in the number of calls to 
\lstinline|ttj|.
All we need to show now is that there are a linear number 
of calls to \lstinline|ttj|.

Because $p$ is the reverse of a GYO reduction order for $Q$,
the following holds from Lemma~\ref{lem:gyo}:
except for the one call to \lstinline|ttj| 
on the root relation (when $i = 1$),
every call to \lstinline|ttj|
has 3 possible outcomes:
\begin{enumerate*}
  \item It outputs a tuple.\label{enum:output}
  \item It backjumps and deletes a tuple from an input relation.\label{enum:backjump}
  \item It recursively calls \lstinline|ttj|.\label{enum:recursive}
\end{enumerate*}
Because the query plan has constant length, 
there can be at most a constant number of recursive calls
to \lstinline|ttj| (case~\ref{enum:recursive}) until we reach cases~\ref{enum:output} or~\ref{enum:backjump}.
Therefore there are at most $O(|Q| + \sum_i |R_i|)$ calls to \lstinline|ttj|,
and the algorithm runs in that time.
\end{proof}
By Theorem~\ref{thm:gyo} every $\alpha$-acyclic query can be GYO-reduced,
therefore \lstinline|ttj| runs in linear time:
\begin{corollary}
For any $\alpha$-acyclic query $Q$, 
there is a plan $p$ such that \lstinline|ttj((),p,1)| computes $Q$ in time $O(|Q| + \sum_i |R_i|)$.
\end{corollary}

\subsection{Comparison with Binary Join and \YA}\label{sec:comparison}
We now prove our claim that,  for any given query plan, \TTJ always matches or outperforms binary hash join.
Because \TTJ and hash join build the exact same set of hash tables,
they share the same cost for hash building.
We therefore focus on the cost of hash lookups which accounts for
the majority of the remaining cost for both algorithms.
The following proofs take advantage of set semantics,
but it is easy to extend the reasoning for bag semantics,
as we can convert a bag into a set by appending a unique
labeled null value to each tuple.
We start with the following observation to relate the run time
of hash join and \TTJ to the set of arguments they are invoked with:
\begin{lemma}\label{lem:distinct}
Both hash join and \TTJ, as defined in Figure~\ref{fig:binary-join} and Figure~\ref{fig:ttj},
are invoked once for each distinct combination of the arguments $(t, p, i)$.
\end{lemma}
\begin{proof}
We prove by induction over the argument $i$.
In the base case when $i$ = 1, both algorithms are invoked once with 
$t = ()$, $i = 1$.
For the inductive step, first consider the hash join algorithm.
For every distinct $t$, \lstinline|join(t, plan, i)| recursively calls 
\lstinline|join(t$\concat$r, plan, i+1)| for every $r \in R_i[k]$.
Since $R_i$ is a set, each $r$ is distinct, so each $t\concat r$ is also distinct.
The same reasoning also applies to \TTJ, as the algorithm will call itself only
for a subset of the tuples in $R_i[k]$.
\end{proof}
In other words, the number of calls to each algorithm is the same
as the number of distinct arguments they are invoked with.
We can now compare the algorithms, by bounding the number of calls
to \TTJ by that of binary join.
\begin{theorem}\label{thm:ttj-bj}
Given a query $Q$ and a plan $p$ for $Q$,
computing $Q$ with \TTJ using $p$ makes at most as many hash lookups
as computing $Q$ with binary join using $p$.
\end{theorem}
\begin{proof}
For clarity we have repeated the lookup \lstinline|R[k]| three 
times in Figure~\ref{fig:ttj},
but we really only need to look up once and save the result 
to a local variable for reuse.
Specifically, a pointer to \lstinline|R[k]| on line~\ref{lst:ttj-lookup}
can be used for the nullness check on the same line,
the loop on line~\ref{lst:ttj-loop},
as well as the deletion%
\footnote{Although the deletion occurs after a recursive function call,
the recursion has constant depth, 
so the pointer dereference has good temporal locality and is likely cheap.}%
on line~\ref{lst:ttj-delete}.
This way, every call to \lstinline|ttj| makes exactly one hash lookup.
Since the binary join algorithm in Figure~\ref{fig:binary-join}
also makes exactly one hash lookup per call,
it is sufficient to bound the number of calls to \lstinline|ttj|
by that of binary join.
By Lemma~\ref{lem:distinct}, it is sufficient to show the distinct arguments
\TTJ is invoked on is a subset of that for binary join.
We prove this by induction over the argument $i$.
When $i = 1$, both \TTJ and binary join are invoked with $t = ()$ and $i = 1$.
For the inductive step,
\lstinline|ttj(t, p, i)| recursively calls \lstinline|ttj(t$\concat$r, p, i+1)|
only if $r \in R_i[k]$,
which implies \lstinline|join(t, p, i)| will also call \lstinline|join(t$\concat$r, p, i+1)|
in binary join.
Therefore, every call to \TTJ is accounted for with a call to binary join.
\end{proof}
Another cost in query execution comes from accessing 
the matching tuples after a successful lookup,
and one can prove that \TTJ accesses no more tuples than binary join,
following the same reasoning as above.
Although backjumping and tuple deletion in \TTJ may in principle
carry an overhead, we will show in Section~\ref{sec:experiments}
that such an overhead is negligible as compared to the cost of hash lookups.
Finally, we note the above proof does not assume an acyclic query.
Section~\ref{sec:cyclic} analyzes the run time of \TTJ on cyclic queries.

While we guarantee \TTJ to always match binary join, we cannot make the same strong claim for \YA.
We will see in Section~\ref{sec:experiments} that \YA
performs better than \TTJ on some queries.
Here we analyze a few extreme cases for some intuition
of how \TTJ compares to \YA :
\begin{example}
  Consider a query where every tuple successfully joins,
  i.e., no lookup fails.
  In this case binary join and \TTJ behaves identically.
  However, \YA spends additional time futilely computing semijoins (without removing any tuple), 
  before following the same execution as binary join and \TTJ
  to produce the output.
\end{example}
\begin{example}
  The other extreme case is when a query has no output, 
  and \YA immediately detects this and stops.
  In fact Example~\ref{ex:main} is such a query:
  all \YA needs to do is the semijoin $T \ltimes U$,
  where it builds a (tiny) hash table for $U$
  and iterate over $T$ once to detect nothing joins.
  In contrast, although \TTJ also runs in linear time,
  it must build the hash table for all of $S$, $T$ and $U$.
\end{example}
%
\section{Optimizations}\label{sec:optimizations} Up until this section \TTJ has
been presented in foundational manner, requiring only minor changes to  \HJ.
Deep consideration of \TTJ reveals many opportunties for enhancement. We present
two direct optimizations of the \TTJ algorithm inspired by research in Constraint
Satisfaction. We name these the {\em deletion
propagation} and {\em no-good list} optimizations.
Deletion propagation is emodied in the TreeTracker algorithm~\cite{DBLP:journals/ai/BayardoM94} and we include it to examine its effectiveness on join evaluation. No-good list is also known as \emph{no-good recording}, which stems from the \emph{constraint learning} method in Constraint Satisfaction~\cite{Dechter2003}.

\paragraph{Deletion Propagation}
Recall that after a lookup failure, a backjump is executed and the offending tuple removed it from its relation based on the corresponding hash key.
There will be executions where all the tuples sharing that hash key are removed. Programatically in line~\ref{lst:ttj-delete}
in Figure~\ref{fig:ttj}
\lstinline|R[k]| becomes empty. If so 
any subsequent lookup, \lstinline|R[k]| will fail.
Instead of continuing execution, as defined so far, we  can immediately backtrack further to the parent of $R$ and {\em propagate} the deletion to $R$'s parent. Said optimization requires 
adding the single following line to the end of Figure~\ref{fig:ttj}:
%
\begin{lstlisting}[numbers=none]
    if R[k] is None & P is not None: return P
\end{lstlisting}
This optimization is not always beneficial. When there are no subsequent lookups to \lstinline|R[k]|
propagating the deletion is unnecessary and carries a small overhead.

\paragraph{No-Good List}
We had remarked in Section~\ref{sec:algorithm} that removing a tuple from the root relation
is pointless, as the same tuple would never be considered again.
However, any tuple in the root relation that shares the same values 
with an offending tuple over the key schema will also fail.
The no-good list optimization comprises adding that set of values to
a blacklist. Each tuple from the root relation is tested for membership in the blacklist.  Since membership in that list mean certain failure no further effort to join that tuple is necessary. This optimization requires three changes to Figure~\ref{fig:ttj}. 

First, the key values must be included as parameters and passed to the parent relation,  line~\ref{lst:backjump}:
\begin{lstlisting}[numbers=none]
    return (P, $\pi_R$(t))
\end{lstlisting}
When catching the backjump (line~\ref{lst:ttj-catch})
at the root relation, those key values are added to the blacklist:
\begin{lstlisting}[numbers=none]
    if result == (R, vals):
      if i == 0: no_good.add(vals) else: R[k].delete(r)
\end{lstlisting}

When iterating over the root relation, (after line~\ref{lst:ttj-loop}), each tuple is tested for membership in the no-good list and if present further processing is skipped.
\begin{lstlisting}[numbers=none]
    if i == 0 & r.matches(no_good): continue
\end{lstlisting}
%
%
The no-good list, $ng$, can be implemented as a hash table.  Suppose the root relation, $R$, has $m$ children $S_1, \dots, S_m$. The lookup key for $ng$ is $\langle S_i, \ell_i \rangle$ where $\ell_i$ is a set containing $\pi_{\keys(R,S_i)}(t)$ (called \emph{no-goods}) for a tuple $t$ from $R$ that caused a lookup failure at $S_i$ for $i \in [m]$. The impact of the no-good list is almost identical to semijoin reduction in \YA. The algorithmic difference is in lieu of a semijoin removing dangling tuples prior to the join, the $R$ tuples are checked against a collection of values accumulated on the fly and at anytime during execution are a subset of the contents of the complementary antijoin. Like \YA itself, the effectiveness of the no-good list depends on how much the argument is reduced and the size of the intermediate result. i.e. the semijoin and join selectivity.  
We demonstrate the trade-off through Star Schema Benchmark in \cref{sec:query-performance}.

\nop{Like deletion propagation, the benefit of a no-good list is data dependent. 
When the list gets big, maintaining and probing it may become more expensive than the lookups saved.}
\section{Empirical Results}\label{sec:experiments}


Since our primary contribution concerns the development of an algorithm that is both asymptotically optimal and is competative in practice w.r.t. wall clock time,
the primary goal of the empirical assessment is to compare the execution time of the algorithms in as controlled of an experiment as possible.  All three algorithms,  \TTJ,  binary hash join, \HJ, and \YA are implemented in the same Java query execution engine written from scratch. We are certain our
algorithm execution measurements do not make calls to methods outside of our execution environment.  Any data structure in our execution environment whose definition is impacted by the definition of a data structure outside of our Java execution environment is treated identically for all three algorithms. Where possible, code is reused across algorithm implementation. The source code of the implementation is available at \url{https://anonymous.4open.science/r/treetracker}.

Remaining aspects of query compilation and and DBMS implementation are ``borrowed'' from other DBMS implementations.
Query plans are an example of borrowing from other DBMS implementations. After loading a benchmark database instance and gathering catalog statistics left-deep linear query plans are determined by SQLite, and bushy plans by PostgreSQL.  The SQL EXPLAIN command elicits the plans from the DBMSs.  SQLite and PostgreSQL were chosen because of the topology of the plans their optimizers generate. The linear time guarantee only holds for left-deep linear plans that are consistant with a GYO reduction order of the query. All the left-deep plans produced by SQLite in our experiments are consistent with the GYO reduction requirement.

\nop{
The study spans the  direct comparison of \TTJ, \HJ and \YA and the impact of the two optimizations, independently and their interaction when both are integrated into the \TTJ algorithm. 
Each algorithm is evaluated on two different kind of plans: left-deep linear plans and bushy plans.
Each left-deep linear plan is produced by SQLite%

whose optimizer is designed to generate such plans,
and can also print out the plans in a machine-readable format.
Every left-deep plan we encountered can be reversed into a GYO order,
thus \TTJ and \YA are guaranteed to run in linear time.
Bushy plans are generated by PostgreSQL,
and each plan is decomposed into a sequence of left-deep linear plans
for execution by different algorithms.}

\textit{Workload.}  Our experments encompass left-deep plans, left-deep plans with optimizations integrated into the \TTJ algorithm, and bushy plans.
Only the acyclic join queries in three benchmarks were evaluated, the Join Ordering Benchmark (JOB)
\cite{Leis2015}, TPC-H \cite{TPC} (scale factor = 1), and the Star Schema Benchmark
(SSB) \cite{ONeil2009a} (scale factor = 1). Also omitted were single-relation
queries, and correlated subqueries. These criteria eliminated only 9 queries, all from TPC-H. Thus, the 113 JOB queries, the 13 SSB queries and
13 out of 22 TPC-H queries were assessed, for a total of 139 queries.

\textit{Environment.} Experiments were conducted on a single logical core of an
AMD Ryzen 9 5900X 12-Core Processor @ 3.7Hz CPU. The computer contained 64 GB of RAM, and a 1TB PCIe NVMe Gen3 M.2 2280 Internal SSD.  
All data structures are allocated from  JVM heap which was set to 20 GB. Since execution was otherwise identical for all algorithms under test, no techniques to reduce the overhead of memory allocation or garbage collection were exploited.
Measurements for each query and algorithm pair were orchestrated by JMH \cite{jmh} configured for 5 warmup forks and
10 measurement forks.  Each of those forks contains 3 warmup
and 5 measurement iterations. 

Direct measurements of PostgreSQL, which can be seen as the control group (not a baseline),  for our implementations are the same as
\cite{Leis2015}. PostgreSQL measurements use an in-memory hash join, indices were dropped and single process execution specified. Thus, we configured PostreSQL such that measurements were made as similar to our Java implementations as we could make possible. A timeout was set to 1 minute. Of all the queries only TPC-H Q20 exceeded the timeout.
 

\subsection{Algorithm Comparison}\label{sec:query-performance}

\begin{figure}
\centering
\begin{subfigure}[t]{.33\linewidth}
\includegraphics[width=\textwidth]{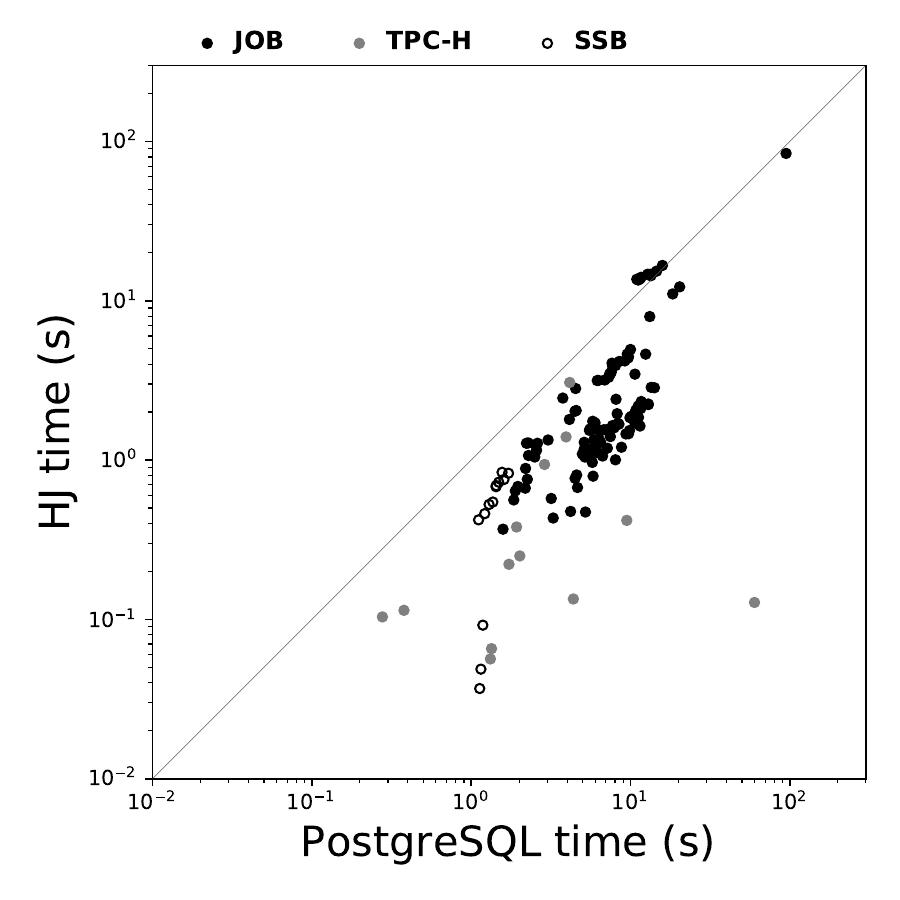}	
\caption{$\HJ$ vs. PostgreSQL}
\label{fig:hj-postgres}
\end{subfigure}%
\begin{subfigure}[t]{.33\linewidth}
\includegraphics[width=\textwidth]{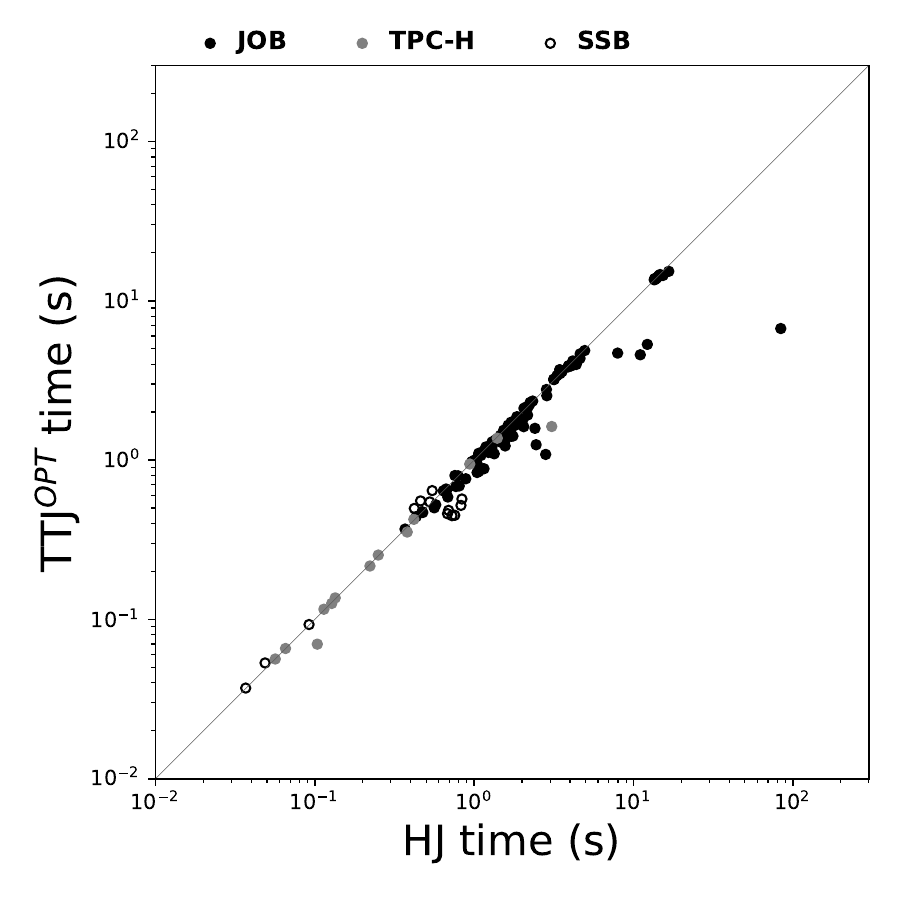}	
\caption{$\TTJ$ vs. $\HJ$}
\label{fig:ttj-hj}
\end{subfigure}%
\begin{subfigure}[t]{.33\linewidth}
\includegraphics[width=\textwidth]{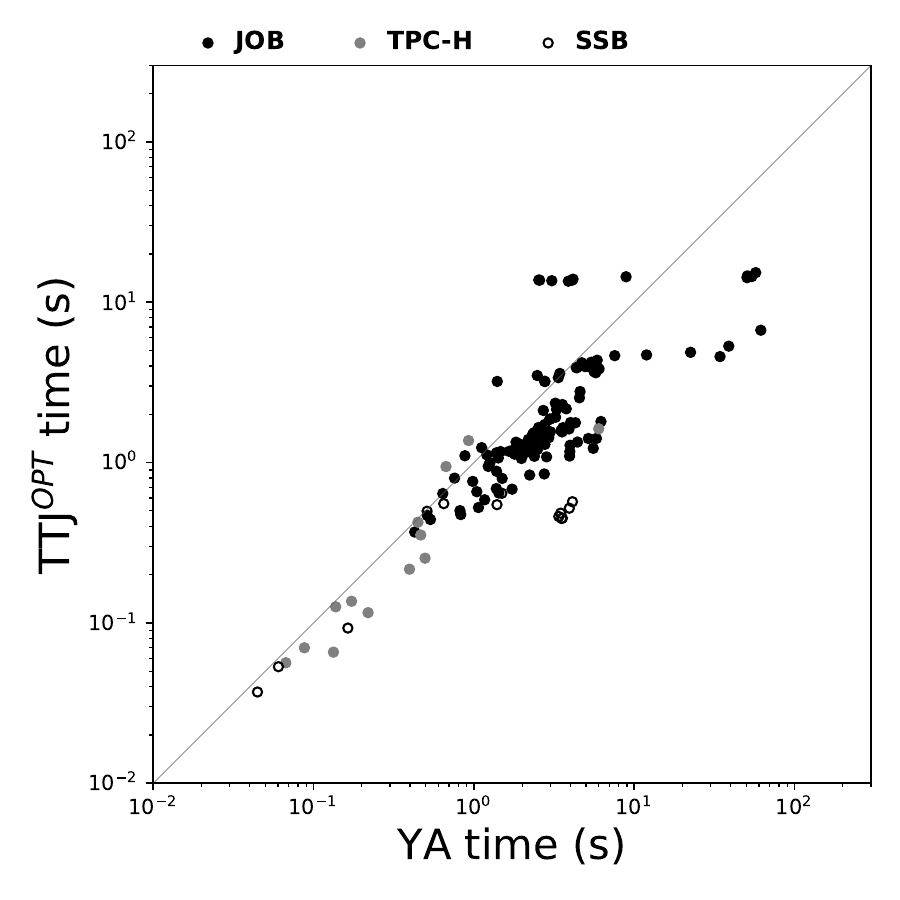}	
\caption{$\TTJ$ vs. $\YA$}
\label{fig:ttj-ya}
\end{subfigure}
\cprotect{\caption{Run time of $\TTJ$, $\HJ$, $\YA$, and PostgreSQL on JOB, TPC-H, and SSB.
Every data point corresponds to a query, whose $x$- and $y$-coordinates correspond to the run time 
of the algorithms under comparison.}}
\label{fig:job}
\end{figure}

\begin{table}
\begin{tabular}{c|c|c|ll|ll} \toprule
\multicolumn{1}{c|}{Baseline} & \multicolumn{1}{l|}{Benchmark} & \multicolumn{1}{c|}{Average} & \multicolumn{2}{c|}{Maximum} & \multicolumn{2}{c}{Mininimum} \\ \midrule
\multirow{3}{*}{Hash Join}  & {JOB}  & $1.11\times$  & $12.6\times$ & (16b) & $0.9\times$  & (11b)  \\
                            & {TPC-H} & $1.09\times$ & $1.9\times$ & (Q9) & $1\times$  & (Q7) \\
                            & {SSB} & $1.15\times$  & $1.7\times$ & (Q2.2) & $0.8\times$  & (Q3.4) \\ \hline

\multirow{3}{*}{\begin{tabular}{@{}c@{}}Yannakakis's\\Algorithm\end{tabular}} & {JOB} & $1.60\times$ & $9.2\times$ & (16b) &
$0.2\times$ & (6a) \\ & {TPC-H} & $1.40\times$ & $3.7\times$ & (Q9) &
$0.7\times$ & (Q7) \\ & {SSB} & $3.16\times$ & $7.9\times$ & (Q2.2) & $1\times$ & (Q3.4) \\ \bottomrule

\end{tabular}
\vspace{1em}
\caption{Speed-up of TreeTracker Join Relative to Hash Join and Yannakakis's Algorithm.}
\label{tab:benchmark-stats}
\end{table}

\cref{fig:job} illustrates our primary results. It contains 3 scatter plots that pairwise compare the execution time of 4 implementations for each query across the 3 benchmarks. First, \cref{fig:hj-postgres} compares the performance of PostgreSQL, using hash joins with our implementation 
of \HJ.   Inspection of the scatterplot shows that with few exceptions the execution time of the same query is less than an order of magnitude apart.  Most points are below the diagonal indicating our implementation is faster than PostreSQL. The shape of the cluster suggests a consistent range in the disparity of execution time.

Faster execution is not surprising. The results of a road race with PostgreSQL are not material to this paper. Our execution environment contains no elements of transaction system overhead or buffer and memory hierarchy management.  PostgreSQL execution time was measured as a control. This first plot establishes that our Java implementation is within range of a commercially used RDBMS and the consistency in the difference of execution speed lends credibility that the emperical results from our execution environment will generalize to commercially deployed RDBMSs.

\cref{fig:ttj-hj} shows, on a per query basis, the relative speed of \TTJ versus \HJ. The visualization in \cref{fig:ttj-hj} reveals that \TTJ is often faster than \HJ, and for just a few queries the execution is slower and when that is the case the performance disadvantage is marginal.  Per \cref{tab:benchmark-stats}, JOB query 11b and SSB query Q3.4, form the worst results for \TTJ are just 10\% and 20\% slower respectively.   The weighted average of \TTJ execution time over the three benchmarks is a hair better than 10\% faster.  More sizable improvements appear in the maximum speed-up results.   We remind the reader the join orders are for plans that were optimized for left-deep linears hash-joins. Below we will return to the question of the upside opportunity for \TTJ execution speed when, in future work, a SQL optimizer includes cost models for \TTJ and the optimization process includes both a choice of join order and a choice of join algorithm.  For instance the detailed examination of each query execution revealed that \TTJ's worst-relative performance, JOB query 11b is due to the inclusion of the no-good list optimization which, often predictably, incures overhead without providing any performance benefit. 

For completeness, performance of \TTJ relative to \YA is presented in \cref{fig:ttj-ya}.  The results exemplify the paradox and challenge of \YA.  On all but 12 queries, \TTJ outperforms \YA,
with average and maximum speedup of 1.4x and 9.2x.
8 of those queries, JOB 6a, 6b, 6c, 6d, 6e, 7b, 12b,
and TPC-H Q7,
exhibit the most significant disadvantage of \TTJ.

Review of the JOB queries reveals a foreseeable cause for \YA execution speed advantage.  The first semijoin removes a large
fraction of tuples from a large relation. For example, the first semijoin for JOB query 6a reduces the largest relation, $\imdbcastinfo$, from  36,000,000
tuples to 486 tuples. That semijoin is executed before
building the hash tables.  Hash table build time for \YA is  499ms. For \TTJ that build time is 13,398ms and by itself comprises 
98\% of the execution time  for \TTJ. 

The basis of TPC-H, Q7's performance results are also due to the impact of the first semijoin, but in a more involved way.  Prior to any join processing a relational select on $\tpchnation$ returns just 1 tuple.  As an argument to the first semijoin,   $\tpchsupplier \ltimes \tpchnation$,  over 90\% of
tuples from $\tpchsupplier$ are removed. Where, in the first example the one semijoin reduction accounted for speed benefit, in this example, by beginning with a single tuple, the entire chain of semijoin reductions resulted in large reductions in the size of the join arguments repeat.

Review of hash table build times for \YA relative to hash table build times for \TTJ and \HJ alone, (these latter two always being equal), accounts for all the speed improvement of \YA compared to the other algorithms.  

The specialized pattern embodied in  SSB, star queries on a star schema, enables a quantitative assessment that may be used in the future by a query optimizer.  Notably a determination if the integration of a no-good list is advantagous. For the special case of
big data queries modeled by SSB, the
performance of \TTJ  is largely
determined by the effectiveness of the no-good list.

Recall the no-good list is a specialization for the leftmost argument of a join plan as hash-joins do not typically create a hash-table for the leftmost argument. The no-good list forms a cache of the tuple key values for the leftmost argument that have been determined to be dangling.  Queries plans for star schema typically start with the fact table as the
leftmost relation in a plan, and provide the key values for a series of joins on dimension tables. Any lookup failure will backjump to the fact table and add to the no-good list. 
Thus the no-good list acts as a filter that prevents any
processing of a fact tuple whose join key values have already been determined to be
fruitless.  

\nop{probably delete:
and if sourced from the right hand argument of a binary join would be deleted from the hash table.  
The hash tables for the righthand join arguments enable \TTJ to remove a tuple as an argument to the query without impacting the base relations. }

Of 13 SSB queries evaluated $\TTJ$ is the fastest algorithm on 6, and, plus or minus,
within 10\% margin of the best algorithm on 10 queries. 
We compared the queries that run relatively slower
in \TTJ (Q1.2, Q3.4, Q4.1, and Q4.3)
with those that run relatively faster 
(Q2.1, Q2.2, Q2.3, Q3.1, Q3.2, and Q3.3), 
and measured the ratio between the
intermediate result size reduction and the size of no-good lists. We determined that for the slower queries, 
each element in the no-good list, on average, reduces the intermediate result size by
182. For the faster queries the average is 318. Although an optimizer is not within the scope of this paper, 
we can conclude that for our testbed  more refined measurements would determine a tipping point value of a selectivity that falls between $1/182$ and $1/318$. Selectivity below the tipping point indicates omitting the no-good list will result in faster query execution and vice versa.

The scatter plot \cref{fig:ttj-probe-hj} compares the number of hash probes for \TTJ vs. \HJ for each query. A small number of the scatter plot points appear on the diagonal, i.e. an equal number of hash probes. The remainder of the points are below the diagonal.
This emperically validates our theorem that 
\TTJ will execute fewer or an equal number of
hash probes as \HJ.

We have not made any claims as to the relative number of hash probes between \TTJ and \YA. Nevertheless we made that measurement. 
\cref{fig:hash-probe-ya}  shows like \HJ, \TTJ makes fewer hash probes for \YA. Yet \YA runs faster for certain queries as hash building,
not probing, sometimes dominates query run time, as we have pointed out
in the analysis of results in~\cref{fig:ttj-ya}.

\nop{
may be suprising. This is likely
The larger gaps between the algorithms in \cref{fig:hash-probe} relative to \cref{fig:job}
is also due to the fact that query execution spends more time building the hash tables
rather than probing them.
In the following, we analyze the results for each benchmark suite in detail.
}

\begin{figure}
\centering	
\begin{subfigure}[t]{.33\linewidth}
\includegraphics[width=\linewidth]{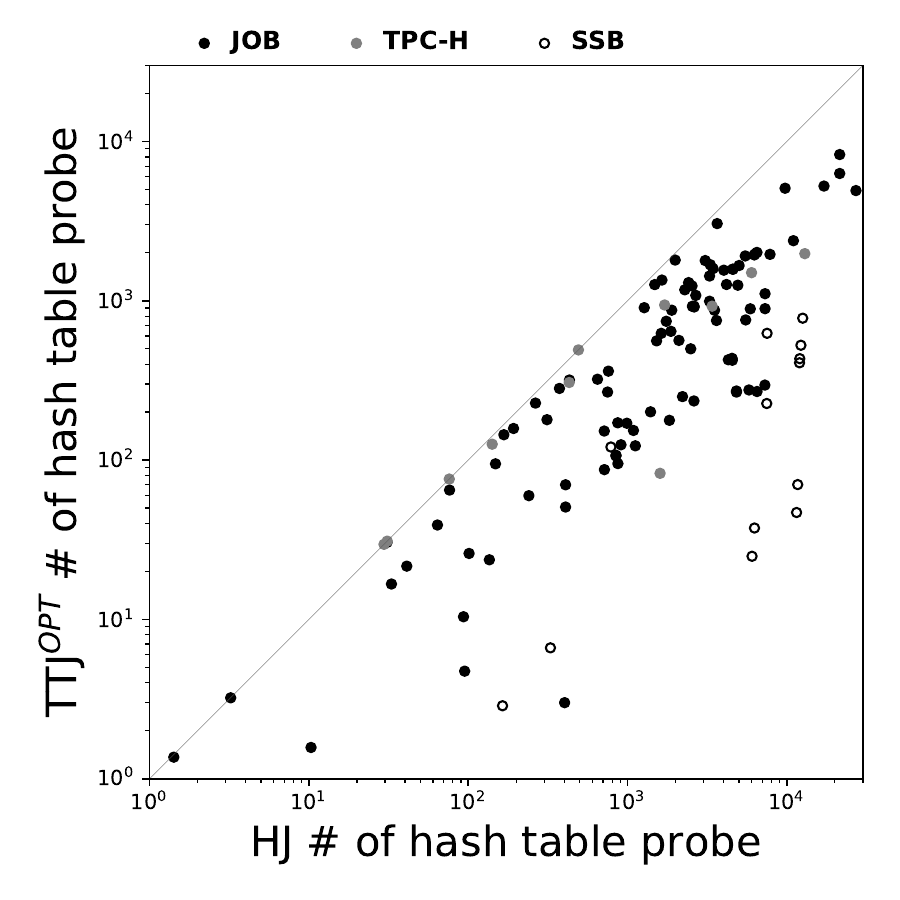}	
\caption{$\TTJ$ vs. $\HJ$}
\label{fig:ttj-probe-hj}
\end{subfigure}
\hspace{4em}
\begin{subfigure}[t]{.33\linewidth}
\includegraphics[width=\linewidth]{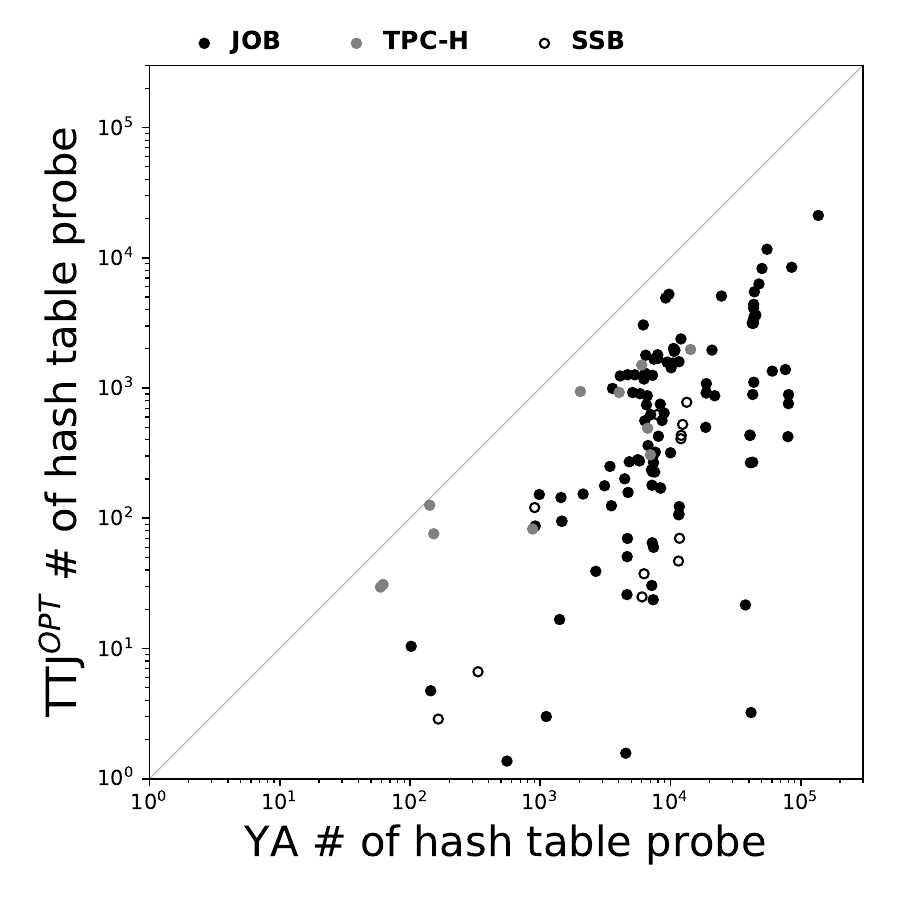}
\caption{$\TTJ$ vs. $\YA$}
\label{fig:hash-probe-ya}
\end{subfigure}
\caption{Number of hash probes in different algorithms.}
\label{fig:hash-probe}
\end{figure}

\nop{
\paragraph{JOB}%
Of all 113 queries, $\TTJ$ is the fastest algorithm
on 74 of them (65\%),
and it is within a 10\% margin of the best algorithm on 100 queries (88\%).
Compared to $\HJ$, the maximum speed-up is $12.6\times$
(16b), the minimum speed-up is $0.9\times$ (11b), and the average speed-up
(geometric mean) is $1.11\times$. Compared to $\YA$, the maximum speed-up is
$9.2\times$ (16b), the minimum speed-up  is $0.2\times$ (6a), and the average
speed-up is $1.60\times$.}%

\subsection{Impact of Optimizations}\label{sec:optimization-impact}
Experiments in this section investigate the impact of the two optimizations
introduced in \cref{sec:optimizations}, no-good list and deletion propagation.
To denote \TTJ with both no-good list ($ng$) and deletion
propagation ($dp$) we write $\TTJNGDP$. To denote \TTJ  with no-good list only and \TTJ  with deletion propagation only we write  $\TTJNG$ and $\TTJDP$ respectively. For this section \TTJ shall mean the algorithm without the optimizations. \cref{fig:optimization-impact}
contains scatter plots that compare the runtime of \TTJ with each of the three possible integrations of the optimizations.

\begin{figure*}[!t]
	\begin{subfigure}{0.33\linewidth}
		\includegraphics[width=\textwidth]{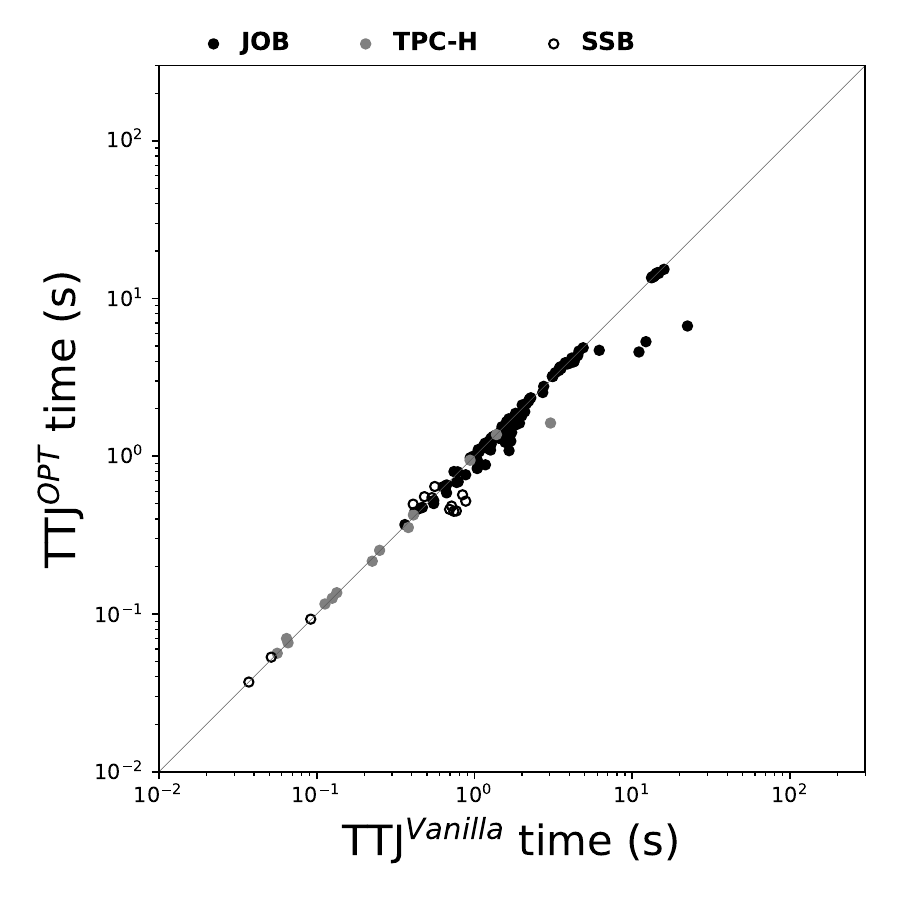}
		\cprotect\caption{Using both optimizations}
		\label{fig:ttj-dp-ng-ttj}
	\end{subfigure}%
	\hfill%
	\begin{subfigure}{0.33\linewidth}
		\includegraphics[width=\textwidth]{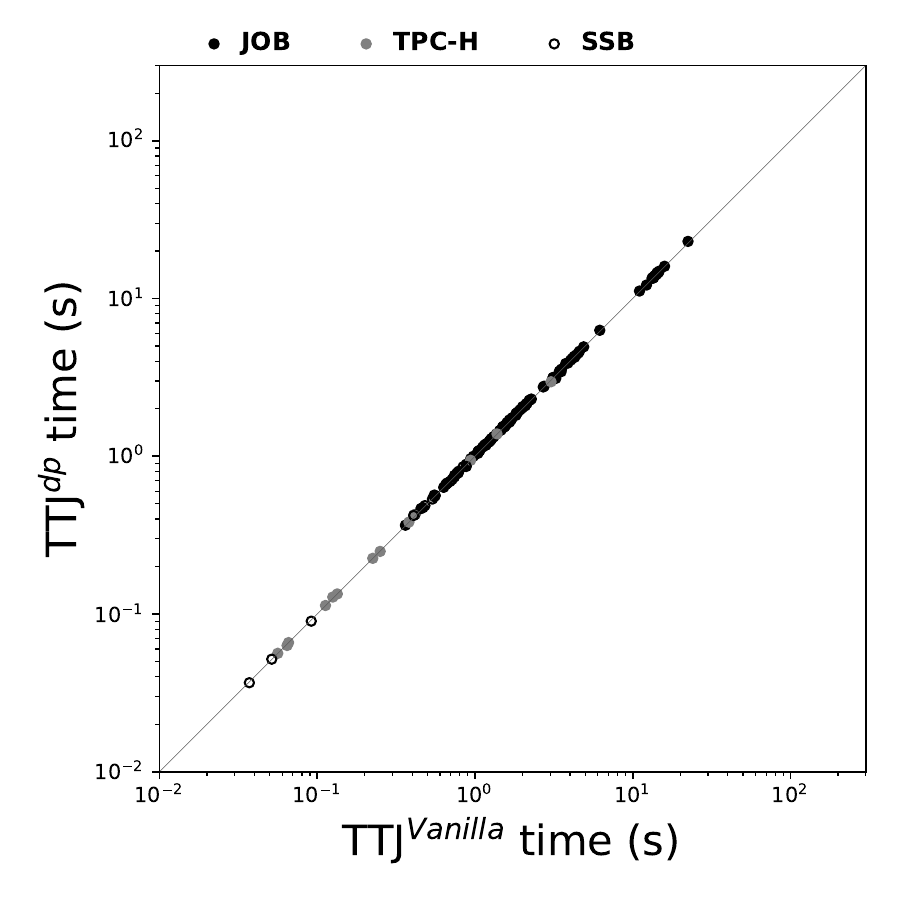}
		\cprotect\caption{Deletion propagation ($dp$)}
		\label{fig:ttj-dp-ttj}
	\end{subfigure}%
	\hfill%
	\begin{subfigure}{0.33\linewidth}
		\includegraphics[width=\textwidth]{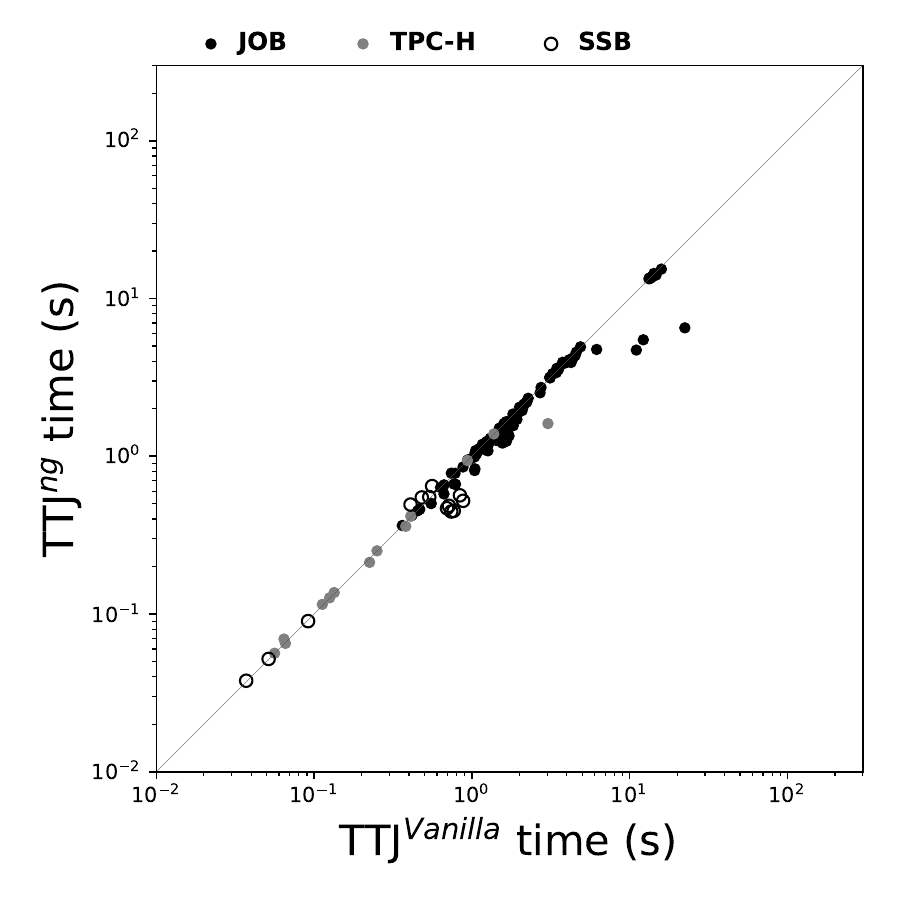}
		\cprotect\caption{No-good list ($ng$)}
		\label{fig:ttj-ng-ttj}
	\end{subfigure}
	\caption{Performance impact of \TTJ optimizations.}
	\label{fig:optimization-impact}
\end{figure*}

Reviewing the behavior of the optimizations independently, we see that deletion propagation 
has little effect on query run time.
This is due to the fundamental differences between constraint satisfaction and query evaluation.
In constraint satisfaction, the {\em modus operandi} is to backtrack upon a failed constraint,
and such failures are common due to the large number of constraints present.
In contrast, in query evaluation the number of constraints is small as the query
is usually much smaller than the input data,
and a large number of tuples will eventually satisfy all constraints and appear in the output.
Indeed, from running the entire benchmark suite, only around 5\% of deletions trigger propagation.

On the other hand, no-good list markedly improves the performance in a number of
queries. There are only a few queries slightly above the diagonal, all from the
SSB benchmark. Recall from our SSB analysis (\cref{sec:query-performance}) that
the effectiveness of the no-good list hinges on the ratio of intermediate result
size reduction to the no-good list size. For star schema queries, this ratio is
dominated by the no-good list's ability to filter tuples from the fact table—the
left-most relation in the join plan. 
Quantifying this, we observe that each no-good element reduces
intermediate results by 318 on average for fast queries (Q2.1, Q2.2, Q2.3, Q3.1,
Q3.2, Q3.3) versus 182 for slow queries (Q1.2, Q3.4, Q4.1, Q4.3). The fast
queries achieve a 75\% higher reduction per element, surpassing the threshold
where benefits outweigh costs. Crucially, this aligns with scenarios where the
join plan structure and backjump dynamics prioritize filtering the left-most
relation—a pattern common in efficient executions.

\subsection{\TTJ on Bushy Plans}\label{sec:bushy-plans}

\begin{figure}
	\begin{subfigure}[t]{.33\linewidth}
		\includegraphics[width=\linewidth]{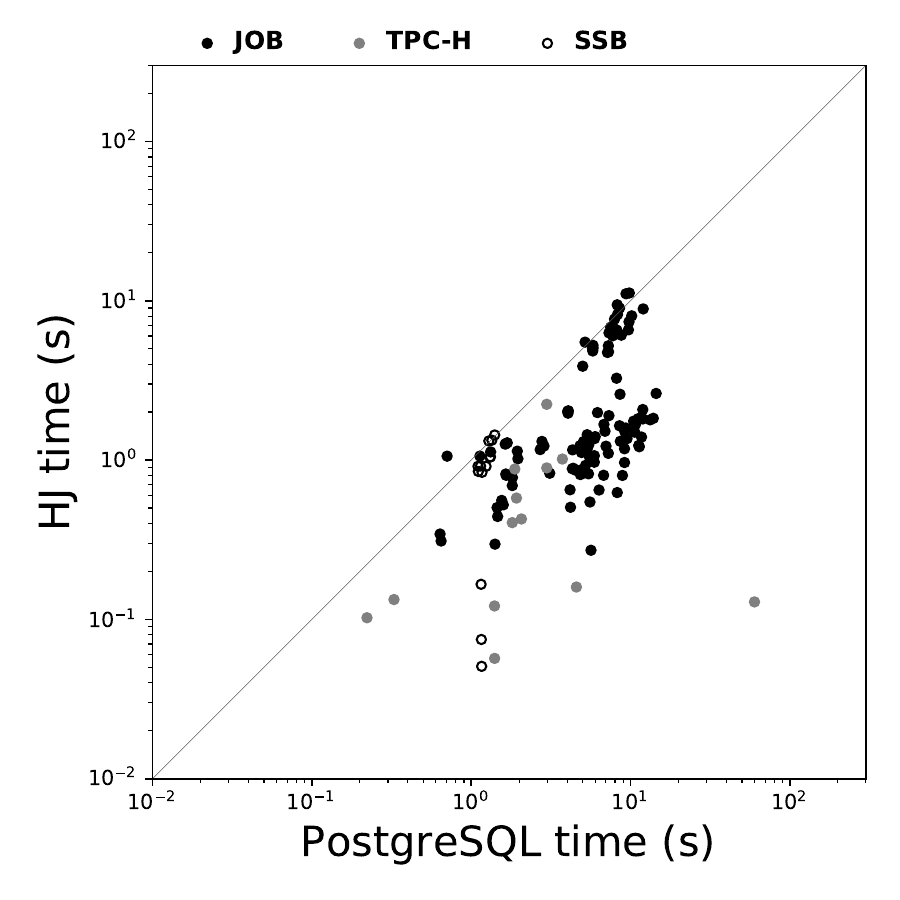}
		\caption{$\HJ$ v.s. PostgreSQL}
		\label{fig:bushy-postgres}
	\end{subfigure}%
	\begin{subfigure}[t]{.33\linewidth}
		\includegraphics[width=\linewidth]{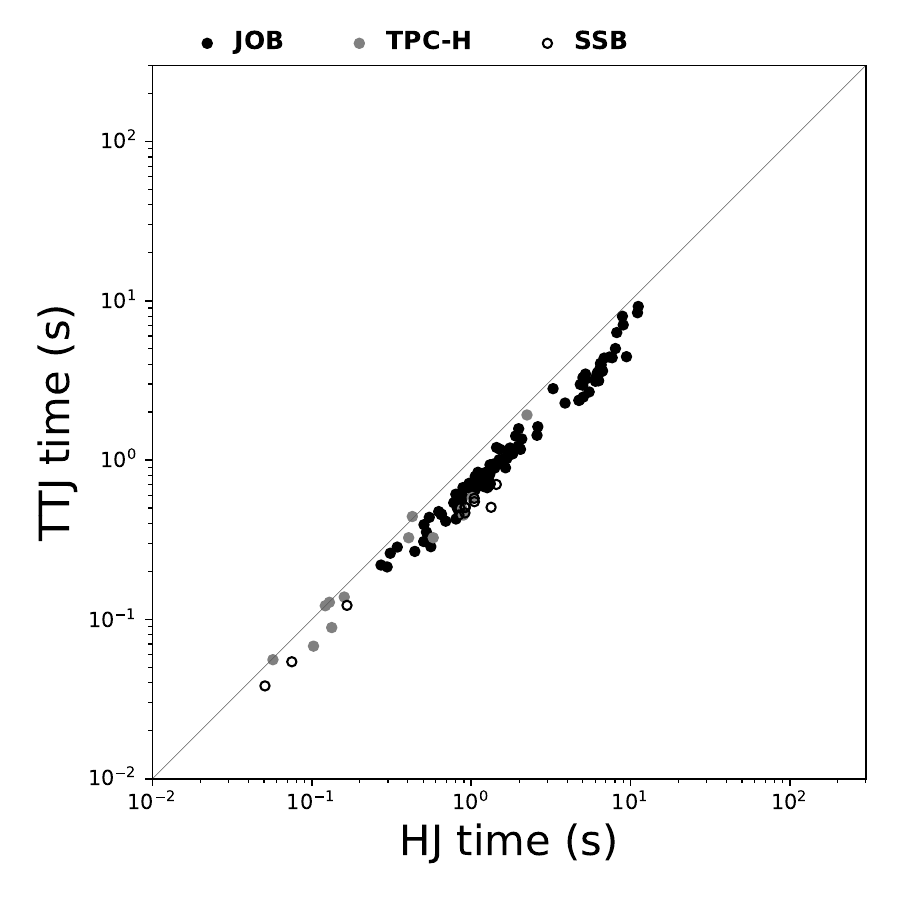}
		\caption{$\TTJ$ v.s. $\HJ$}
		\label{fig:bushy-hj}
	\end{subfigure}%
	\begin{subfigure}[t]{.33\linewidth}
		\includegraphics[width=\linewidth]{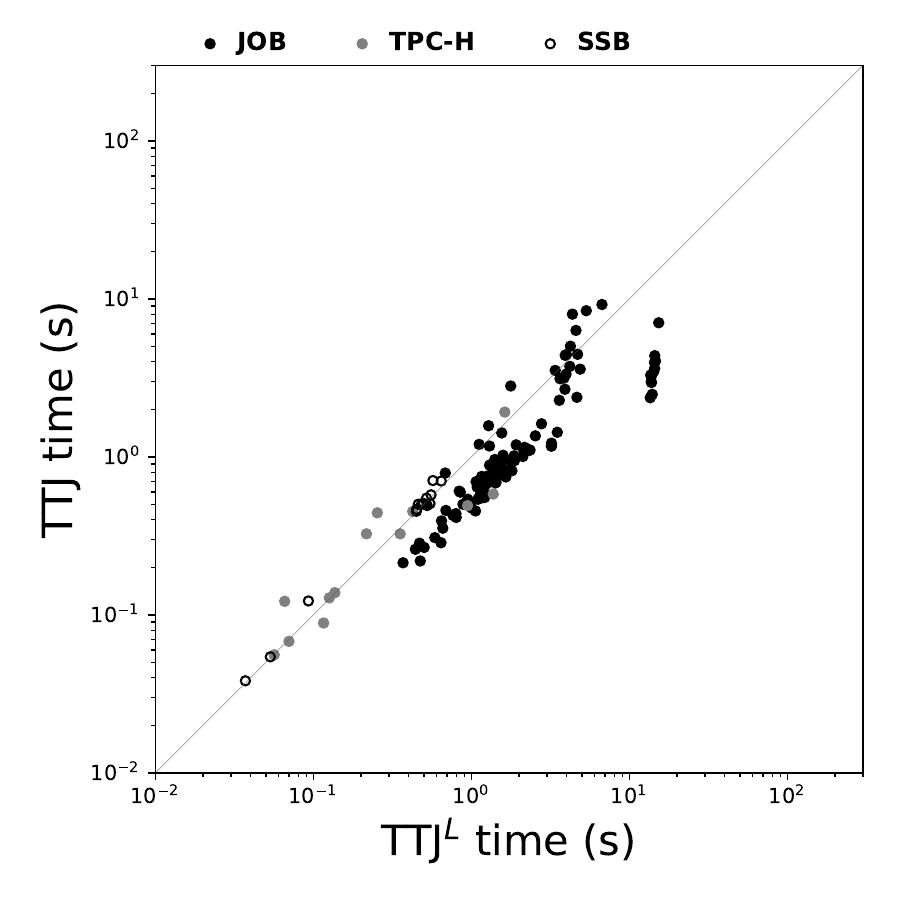}
		\caption{$\TTJ$ v.s. $\TTJLinear$}
		\label{fig:bushy-ttjl}
	\end{subfigure}
	\cprotect{\caption{%
	Run time comparison of $\TTJ$, $\TTJLinear$, $\HJ$ and PostgreSQL.
	$\TTJLinear$ uses left-deep linear plans generated by SQLite (same as earlier experiments),
	while all other algorithms use bushy plans generated by PostgreSQL.}}
	\label{fig:job-bushy}
\end{figure}

\nop{\begin{figure}
	\begin{adjustbox}{addcode={\begin{minipage}{\width}}{\caption{%
						speed-up of  $\TTJ$ (under PostgreSQL plans) and $\TTJLinear$ (under linear plans) over PostgreSQL on all 113 JOB queries using native PostgreSQL plans.\label{fig:job-bushy}}
  \rw{Can you rotate this back, but break into more rows to fit the page?}
					\end{minipage}},rotate=90,center}
		\includegraphics[width=\textheight]{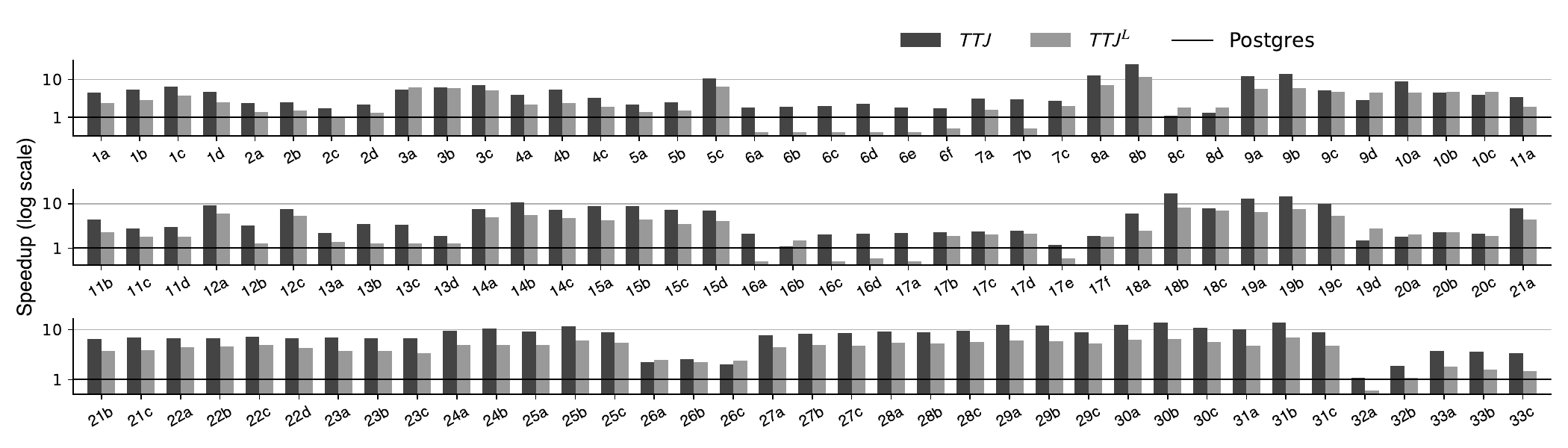}%
	\end{adjustbox}
\end{figure}}

Since every plan in \cref{sec:query-performance} is compatible with a GYO
reduction order, the runtime of \TTJ is guaranteed to be linear. However, given
the abundance of bushy plans (all the native PostgreSQL plans we used here are
bushy), a natural question to ask is whether \TTJ can still provide reasonably
good performance despite the loss of linear runtime guarantee. The results in
this section give an affirmative answer.

\cref{fig:bushy-postgres,fig:bushy-hj} compare the run time of $\TTJ$, $\HJ$,
and PostgreSQL using bushy plans produced by PostgreSQL,
and \cref{fig:bushy-ttjl} compares the run time of $\TTJ$ using bushy plans with 
the same algorithm using left-deep linear plans produced by SQLite, denoted by $\TTJLinear$.
\cref{fig:bushy-postgres} shows our \HJ baseline remains competitive with PostgreSQL under bushy plans.
$\TTJ$ is faster than $\HJ$ on all 113 JOB queries,
and faster than $\TTJLinear$ on 101 (89\%) of them.
Compared to $\TTJLinear$, the maximum
speed-up is $5.7\times$ (7b), the minimum speed-up is $0.5\times$ (19d), and the
average speed-up (geometric mean) is $1.75\times$.  Compared to \HJ, the maximum
speed-up is $2.1\times$ (13d), the minimum speed-up is $1.1\times$ (19d), and
the average speed-up (geometric mean) is $1.56\times$. 
From the figure, we observe that the materialization of intermediate results due
to bushy plans does not degrade \TTJ performance; in fact, \TTJ performs
much better than itself on linear plans in most cases due to the fact that
intermediate results generated in bushy plans are usually smaller than some of
the largest input relations in JOB, which allows \TTJ to spend less time building hash
tables. As a result, the saving in join computation becomes more salient than that
over left-deep plans. Furthermore, the performance improvement using \TTJ on each
individual linear plan has the compound effect that contribute to the overall
performance improvement of the queries.
However, we do observe that some of the queries still have better performance
under linear plans than bushy plans such as 8c and 16b. This indicates that
linear time guarantee is still meaningful for query performance and optimization
is still necessary to bring out the best performance of \TTJ (e.g., decide which
plan shape to use). 
An important topic for future work is to optimize bushy plans that also 
have the guarantee of optimality.
\section{Cyclic Queries}\label{sec:cyclic}
The \TTJ algorithm as defined in Figure~\ref{fig:ttj} supports both
acyclic and cyclic queries.
The guarantee to match or outperform binary join
(Theorem~\ref{thm:ttj-bj}) also holds for cyclic queries.
However, the linear-time guarantee only applies to acyclic queries.
To analyze the run time of \TTJ on cyclic queries,
we introduce a new method called {\em tree convolution} to 
break down a cyclic query in to acyclic parts. 
The next example illustrates the intuition behind tree convolution.
%
\begin{example}\label{ex:box}
Consider the following query whose query graph is shown in Figure~\ref{fig:box-query}:
\begin{align*}
Q_\boxtimes \cd & R_1(x_1, x_2) \bowtie R_2(x_2, x_3) \bowtie R_3(x_3, x_4) \bowtie R_4(x_4, x_1) \bowtie \\
& S_1(x_1, y) \bowtie S_2(x_2, y) \bowtie S_3(x_3, y) \bowtie S_4(x_4, y)
\end{align*}
In the query graph, each node is a variable,
and there is an edge between two nodes if the corresponding variables appear in the same atom,
for example the edge $R_1$ between $x_1$ and $x_2$ corresponds to the atom $R_1(x_1, x_2)$.
Clearly the query is cyclic.
However, we can compute the query by breaking it down into two acyclic steps:
First, we compute the join $S_1 \bowtie S_2 \bowtie S_3 \bowtie S_4$
and store the result in a temporary relation $S$.
Then we compute the final result with $R_1 \bowtie R_2 \bowtie R_3 \bowtie R_4 \bowtie S$.
Any acyclic join algorithm can be used to compute each step,
and, as we will show later, using \TTJ can avoid materializing the intermediate result $S$.
The total run time is therefore $O(|\IN| + |\OUT| + |S|)$.
\end{example}

\begin{figure}
\begin{subfigure}[t]{0.3\linewidth}
\centering
\begin{tikzpicture}[scale=2]

    \node[draw, circle, fill=black, inner sep=1.5pt, label=above left:$x_1$] (A) at (0,1) {};
    \node[draw, circle, fill=black, inner sep=1.5pt, label=above right:$x_2$] (B) at (1,1) {};
    \node[draw, circle, fill=black, inner sep=1.5pt, label=below right:$x_3$] (C) at (1,0) {};
    \node[draw, circle, fill=black, inner sep=1.5pt, label=below left:$x_4$] (D) at (0,0) {};
    \node[draw, circle, fill=black, inner sep=1.5pt, label=below:$y$] (E) at (0.5,0.5) {};
    
    \draw (A) -- (B) node[midway, above] {$R_1$};
    \draw (B) -- (C) node[midway, right] {$R_2$};
    \draw (C) -- (D) node[midway, below] {$R_3$};
    \draw (D) -- (A) node[midway, left] {$R_4$};
    
    \draw (E) -- (A) node[midway, above] {$S_1$};
    \draw (E) -- (B) node[midway, above] {$S_2$};
    \draw (E) -- (C) node[midway, below] {$S_3$};
    \draw (E) -- (D) node[midway, below] {$S_4$};
    
    \end{tikzpicture}
\caption{Quey graph of $Q_\boxtimes$.}
\label{fig:box-query}
\end{subfigure}%
\begin{subfigure}[t]{0.3\linewidth}
\centering
\includegraphics{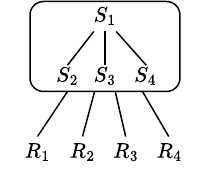}
\caption{A tree convolution of $Q_\boxtimes$.}
\label{fig:box-conv-1}
\end{subfigure}%
\begin{subfigure}[t]{0.4\linewidth}
\centering
\includegraphics{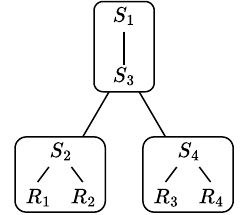}
\caption{Anothre tree convolution of $Q_\boxtimes$.}
\label{fig:box-conv-2}
\end{subfigure}
\caption{Query graph and tree convolutions of $Q_\boxtimes$ in Example~\ref{ex:box}.}
\label{fig:box-query-conv}
\end{figure}

To formally define tree convolutions, we first identify a (full conjunctive) query with
the set of atoms in its body,
and define a {\em subquery} of a query $Q$ as a subset of the atoms in $Q$.
Then, a tree convolution for a query is a nested tree, defined recursively as follows:
\begin{definition}\label{def:tree-conv}
A tree convolution of a query $Q$, written $\tconv(Q)$, is a tree where each node
is either an atom in $Q$, or a tree convolution of a subquery of $Q$.
Each atom in $Q$ appears exactly once anywhere in $\tconv(Q)$,
and every tree $T$ in $\tconv(Q)$ forms a join tree of the corresponding (sub-)query,
after replacing each non-atom node $v\in T$ with 
a fresh atom over all variables in $v$.
\end{definition}

Figure~\ref{fig:box-conv-1} shows a tree convolution corresponding to
the computation of $Q_\boxtimes$ in Example~\ref{ex:box}:
first join together the $S$ relations, then join the result with the $R$ relations.
Figure~\ref{fig:box-conv-2} shows a different convolution, where we first compute
three acyclic subqueries, then join together the results.
Another convolution with four nesting levels is shown in Figure~\ref{fig:bin-conv}.

We can compute any cyclic query using a tree convolution.
Starting from the most deeply nested trees, we run any acyclic join
algorithm to materialize intermediate results,
until we reach the top-level tree that produces the final output.
However, \TTJ can avoid the expensive materialization by using
a special kind of tree convolution called {\em rooted convolution}.
%
\begin{definition}\label{def:rooted-conv}
A convolution is {\em rooted} if nested convolutions only 
appear at the root of each tree.
\end{definition}


%
The convolution in Figure~\ref{fig:box-conv-1} is rooted,
while the one in Figure~\ref{fig:box-conv-2} is not.
We can generate a plan $p$ for \TTJ by traversing a rooted convolution
inside-out:
starting with the most deeply nested tree,
initialize $p$ with the reverse of the GYO-reduction order 
of this tree;
then, as we go up each level, append the reverse of the GYO-reduction order
to the end of $p$.
\begin{example}\label{ex:rooted-plan}
The rooted convolution in Figure~\ref{fig:box-conv-1}
generates the plan $[S_1, S_2, S_3, S_4, R_1, R_2, R_3, R_4]$.
\end{example}

A small adjustment to the \TTJ algorithm is necessary to fully exploit rooted convolutions.
If we execute \TTJ as-is using the plan in Example~\ref{ex:rooted-plan},
none of the relations $R_1, \ldots, R_4$ have a parent in the plan,
yet we need to compute the join $S \bowtie R_1 \bowtie \cdots \bowtie R_4$
in time $O(|S| + \sum_i |R_i| + |Q|)$, where $S$ is the join of $S_1, \ldots, S_4$.
We therefore introduce additional backjumps from each $R_i$ to $S_4$, but without
deleting any tuple from $S_4$.
This is achieved by defining the \lstinline|parent|
function to work over rooted convolutions:
given a tree convolution $C$ and a relation $R$,
if the parent node of $R$ in some tree of $C$
is an atom $P$, then assign $P$ as the parent of $R$;
otherwise if the parent node is a nested tree,
then assign the last relation in that tree (i.e. the first relation 
in its GYO-reduction order) as the parent of $R$.
For example, the parent of each $R_i$ in Example~\ref{ex:rooted-plan}
is $S_4$.



Finally, when backjumping to a parent in a nested convolution,
we do not delete any tuple from that parent.
We are now ready to analyze the run time of \TTJ on cyclic queries.


Given a rooted tree convolution $\tconv(Q)$, we generate a query plan $p$ as follows. Suppose $\tconv(Q)$ consists of $m$ nested trees $T_1, \dots, T_m$
where $T_1$ is the most deeply nested tree and $T_m$ is the outermost tree.
Let $p_1, \dots, p_m$ be the plans corresponding to $T_1, \dots, T_m$.
Then, each $p_i$ is a plan that corresponds to the reverse of a GYO-reduction order of $T_i$, where the first relation in $p_i$ is result of $p_{i-1}$ for $i>1$. The outermost plan $p_m$ then computes the final result of $Q$.

\begin{proposition}{\label{prop:cyclic-backjumping}}
	During $\ttj$ execution on a given rooted convolution $\tconv$ of a query $Q$, if a lookup fails at $R$ that belongs to $p_i$ but not in $p_{i-1}$ ($i \ge 2$), $\ttj$ either backjumps to an atom that is in $p_i$ but not in $p_{i-1}$ or backjumps to the last atom of $p_{i-1}$.
\end{proposition}

\begin{proof}
	For any $p_i$ with $i \ge 2$, since $p_i$ is the reverse of a GYO-reduction order of the $i$-th tree of $\tconv$ and by the modified \lstinline|parent| function, every relation in $p_i$ but not in $p_{i-1}$ has parent. Furthermore, if lookup fails at $R$ that is in $p_i$ but not in $p_{i-1}$, \lstinline|ttj| backjumps to $R$'s parent. Then, the result follows by the definition of the modified \lstinline|parent| function.
\end{proof}

In $p_1$, if the parent of a relation is the first relation of $p_1$, \lstinline|ttj| backjumps to the first relation. We can treat the first relation of $p_1$ as $p_0$ (i.e., the most deeply nested tree in $\tconv$ is now a node of an atom of $Q$) and the relation is also the last atom of $p_0$. Therefore, we can remove the restriction of $i \ge 2$ in \cref{prop:cyclic-backjumping}. In the following proof, we reference the \cref{prop:cyclic-backjumping} with the understanding that it holds for $i \ge 1$.

\begin{theorem}[restate=ttjcyclic,name=,]\label{thm:ttj-cyclic}
	Given a rooted convolution $\tconv(Q)$ of $Q$,
	there is a plan $p$ such that
	\TTJ runs in time $O(|\IN| + |\OUT| + \sum_i |S_i|)$ on $p$
	where $|S_i|$ is the size of the join of all relations in the $i$-th tree of $\tconv(Q)$.
\end{theorem}
\begin{proof}
	Since there is no change to \lstinline|ttj| except using the modified
	\lstinline|parent| function for $\tconv$, like Proof of \cref{thm:linear}, we only
	need to show the algorithm makes $O(|\IN| + |\OUT| + \sum_i |S_i|)$ number of calls to
	\lstinline|ttj|.
	
	Excep for the initial call to \lstinline|ttj| with \lstinline|ttj((), p, 1)|, every call to \lstinline|ttj| has three possible outcomes: 
	\begin{enumerate*}
		\item It outputs a tuple.\label{enum:output-a}
		\item It backjumps and possibly deletes a tuple from an input relation.\label{enum:backjump-a}
		\item It recursively calls \lstinline|ttj|.\label{enum:recursive-a}
	\end{enumerate*}
	Because the query plan has constant length, 
	there can be at most a constant number of recursive calls
	to \lstinline|ttj| (case~\ref{enum:recursive-a}) until we reach cases~\ref{enum:output-a} or~\ref{enum:backjump-a}.
	There can be $O(|Q|)$ \lstinline|ttj| calls for case~\ref{enum:output-a}. 
	Since lookup cannot fail at the first relation of $p$, the relation that lookup fails at is in some $p_i$ but not in $p_{i-1}$. Let $R$ be a relation that a lookup fails at. By \cref{prop:cyclic-backjumping}, there can be two cases on where \lstinline|ttj| backjumps to. If \lstinline|ttj| backjumps to an atom that is also in $p_i$ but not in $p_{i-1}$, a tuple is deleted. This case can happen $O(|\IN|)$ times. If \lstinline|ttj| backjumps to the last atom of $p_{i-1}$, since \lstinline|ttj| works no different than binary join from this moment until next lookup failure, this can happen $O(|S_{i-1}|)$ times; the tuples computed at the last atom of $p_{i-1}$ is $S_{i-1}$.
	Therefore, given $1 \le i \le m$, there are at most $O(|\IN| + |\OUT| + \sum_{i=1}^{m-1} |S_i|)$ calls to \lstinline|ttj|, and the algorithm runs in that time.
\end{proof}

\nop{\begin{theorem}[restate=ttjcyclic,name=,]\label{thm:ttj-cyclic}
Given a rooted convolution $\tconv(Q)$ of $Q$,
there is a plan $p$ such that
\TTJ runs in time $O(|\IN| + |\OUT| + \sum_i |S_i|)$ on $p$
where $|S_i|$ is the size of the join of all relations in the $i$-th tree of $\tconv(Q)$.
\end{theorem}
%
In other words, \TTJ runs in time linear to the input and output size, 
plus the total size of intermediate results.
We sketch the proof of this claim as follows. The full proof is in \cref{sec:full-thm-ttj-cyclic}.
\begin{proof}[Proof sketch]
Consider the prefix $p'$ of length $l$ of the plan $p$ generated by the rooted convolution $\tconv(Q)$,
corresponding to the most deeply nested tree.
During execution, the algorithm makes exactly the same set of calls to \lstinline|ttj(t, p, i)|
when $i \leq l$ as if it were computing the subquery corresponding to $p'$.
This is because any backjumps to $p'$ from a latter part of the plan never deletes any tuple, 
and the algorithm simply continues to produce the next result.
The absence of tuple deletion also does not affect the run time of the next convolution level,
as deleting from the root relation is a no-op as we have noted in Example~\ref{ex:ttj-unroll}.
The same argument applies to any nesting level in $\tconv(Q)$,
so \TTJ runs in linear time for every tree in $\tconv(Q)$.
\end{proof}}

We conclude this section by noting that tree convolution generalizes several exisiting ideas in
databases and constraint satisfaction.
First, classic binary join plans are a special case, where every tree in the tree convolution
is of size 2.
For example, the convolution in Figure~\ref{fig:bin-conv} corresponds to the
binary join plan in Figure~\ref{fig:bin-plan}.
In other words, traditional hash join can be thought of as computing a convolution
one binary join at a time, 
and we have generalized this to computing a multi-way acyclic join at a time.
In a similar way, rooted convolutions generalize left-deep linear plans,
as the top-half of Figure~\ref{fig:bin-conv} corresponds to the left-half in Figure~\ref{fig:bin-plan}.
Second, in constraint satisfaction a {\em cycle cut set} is a subset of
the constraints whose exclusion makes the constraint problem acyclic.
In database terms, it is a subset of the atoms of $Q$ whose removal leaves
an acyclic subquery.
Tree convolution generalizes cycle cut sets in the sense that it 
``cuts'' the query in multiple rounds, with each round producing
acyclic subqueries that can be computed in linear time.
Finally, many of the properties of tree convolution are shared with tree decomposition of hypergraphs~\cite{Gottlob2016}
in database theory. In fact both tree convolutions in \cref{fig:box-query-conv} can also
be seen as tree decompositions: \cref{fig:box-conv-1} is a tree decomposition with 5 bags, one
containing $\{S_1, S_2, S_3, S_4\}$, and one for each $R_i$; 
each box in \cref{fig:box-conv-2} forms a bag
in the corresponding tree decomposition. These are then also examples of where algorithms using tree
decompositions require materializing the join result of each bag. But given a rooted tree convolution, \cref{fig:box-conv-1}, $\TTJ$
requires no materialization.

\begin{figure}
\begin{subfigure}[t]{0.5\linewidth}
\centering
\includegraphics{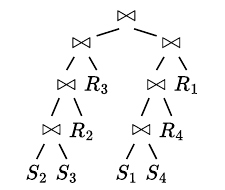}
\caption{A binary join plan for $Q_\boxtimes$.}
\label{fig:bin-plan}
\end{subfigure}%
\begin{subfigure}[t]{0.5\linewidth}
\centering
\includegraphics{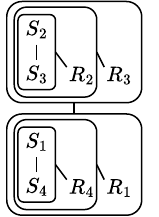}
\caption{A tree convolution for $Q_\boxtimes$.}
\label{fig:bin-conv}
\end{subfigure}
\caption{A binary join plan for the query in Example~\ref{ex:box} and the corresponding tree convolution.}
\label{fig:bin-plan-conv}
\end{figure}
\section{Future Work and Conclusion}\label{sec:conclusion}
In this paper we have proposed our new join algorithm, TreeTracker Join (\TTJ).
The algorithm runs in time $O(|\IN| + |\OUT|)$ on acyclic queries,
and guarantees to make no more hash probes than binary hash join 
on the same query plan.
We have shown empirically that \TTJ is competitive with binary hash 
join and Yannakakis's algorithm.

Although our implementation already beats PostgreSQL in our experiments, 
challenges remain for \TTJ to compete with highly optimized systems.
Decades of research on binary join has produced effective techniques 
like column-oriented storage, vectorized execution, and parallel execution, 
just to name a few.
Future research should investigate how to adapt these techniques to \TTJ.

Another avenue for future work is to develop a dedicated query optimizer for \TTJ.
As this paper's focus is on algorithm-level comparison, 
we have opted to reuse existing systems to produce binary hash join plans,
which are then executed using \TTJ.
Tailoring the optimizer to \TTJ may yield plans with better performance.
For instance, estimating the number of hash probe failures instead of
intermediate result sizes shall more accurately model the execution cost of \TTJ.
On the other hand, extending \TTJ to also guarantee a linear time complexity
on bushy plans is also an interesting challenge.
Our theoretical analyses of \TTJ reveal a close connection between
GYO-reduction orders and left-deep linear plans, both of which are total orders.
Since bushy plans and join trees both define partial orders,
we conjecture there exists a algorithm that runs in linear time on any bushy plan
that corresponds to a join tree,
with the same guarantee of matching binary hash join on the same plan.

Finally, our experiments focus on acyclic queries due to their prevalence 
in traditional workloads.
However, with the rise of graph databases practitioners begin to encounter more and more 
cyclic queries.
Additional research on \TTJ for cyclic queries, 
both in terms of practical performance and theoretical guarantees,
will be very valuable.
Some open problems include:
Given any hypergraph, what is the minimum nesting depth of any tree convolution?
How does this number related to other measures like various notions of hypergraph widths?
And what is the complexity of finding the optimal tree convolution given a cost function?
Answering theses questions will aid the development of a query optimizer for \TTJ
on cyclic queries.


\bibliographystyle{ACM-Reference-Format}
\bibliography{references}


\begin{thebibliography}{24}


\ifx \showCODEN    \undefined \def \showCODEN     #1{\unskip}     \fi
\ifx \showISBNx    \undefined \def \showISBNx     #1{\unskip}     \fi
\ifx \showISBNxiii \undefined \def \showISBNxiii  #1{\unskip}     \fi
\ifx \showISSN     \undefined \def \showISSN      #1{\unskip}     \fi
\ifx \showLCCN     \undefined \def \showLCCN      #1{\unskip}     \fi
\ifx \shownote     \undefined \def \shownote      #1{#1}          \fi
\ifx \showarticletitle \undefined \def \showarticletitle #1{#1}   \fi
\ifx \showURL      \undefined \def \showURL       {\relax}        \fi
\providecommand\bibfield[2]{#2}
\providecommand\bibinfo[2]{#2}
\providecommand\natexlab[1]{#1}
\providecommand\showeprint[2][]{arXiv:#2}

\bibitem[jmh({[n.\,d.]})]%
        {jmh}
 \bibinfo{year}{[n.\,d.]}\natexlab{}.
\newblock \bibinfo{title}{{Java Microbenchmark Harness (JMH)}}.
\newblock
\urldef\tempurl%
\url{https://github.com/openjdk/jmh}
\showURL{%
\tempurl}


\bibitem[Bagan et~al\mbox{.}(2007a)]%
        {Bagan2007OnAC}
\bibfield{author}{\bibinfo{person}{Guillaume Bagan}, \bibinfo{person}{Arnaud Durand}, {and} \bibinfo{person}{Etienne Grandjean}.} \bibinfo{year}{2007}\natexlab{a}.
\newblock \showarticletitle{On Acyclic Conjunctive Queries and Constant Delay Enumeration}. In \bibinfo{booktitle}{\emph{Annual Conference for Computer Science Logic}}.
\newblock
\urldef\tempurl%
\url{https://api.semanticscholar.org/CorpusID:15398587}
\showURL{%
\tempurl}


\bibitem[Bagan et~al\mbox{.}(2007b)]%
        {DBLP:conf/csl/BaganDG07}
\bibfield{author}{\bibinfo{person}{Guillaume Bagan}, \bibinfo{person}{Arnaud Durand}, {and} \bibinfo{person}{Etienne Grandjean}.} \bibinfo{year}{2007}\natexlab{b}.
\newblock \showarticletitle{On Acyclic Conjunctive Queries and Constant Delay Enumeration}. In \bibinfo{booktitle}{\emph{Computer Science Logic, 21st International Workshop, {CSL} 2007, 16th Annual Conference of the EACSL, Lausanne, Switzerland, September 11-15, 2007, Proceedings}} \emph{(\bibinfo{series}{Lecture Notes in Computer Science}, Vol.~\bibinfo{volume}{4646})}, \bibfield{editor}{\bibinfo{person}{Jacques Duparc} {and} \bibinfo{person}{Thomas~A. Henzinger}} (Eds.). \bibinfo{publisher}{Springer}, \bibinfo{pages}{208--222}.
\newblock
\href{https://doi.org/10.1007/978-3-540-74915-8\_18}{doi:\nolinkurl{10.1007/978-3-540-74915-8\_18}}


\bibitem[Bekkers et~al\mbox{.}(2024)]%
        {DBLP:journals/corr/abs-2411-04042}
\bibfield{author}{\bibinfo{person}{Liese Bekkers}, \bibinfo{person}{Frank Neven}, \bibinfo{person}{Stijn Vansummeren}, {and} \bibinfo{person}{Yisu~Remy Wang}.} \bibinfo{year}{2024}\natexlab{}.
\newblock \showarticletitle{Instance-Optimal Acyclic Join Processing Without Regret: Engineering the Yannakakis Algorithm in Column Stores}.
\newblock \bibinfo{journal}{\emph{CoRR}}  \bibinfo{volume}{abs/2411.04042} (\bibinfo{year}{2024}).
\newblock
\href{https://doi.org/10.48550/ARXIV.2411.04042}{doi:\nolinkurl{10.48550/ARXIV.2411.04042}}
\showeprint[arXiv]{2411.04042}


\bibitem[Birler et~al\mbox{.}(2024)]%
        {DBLP:journals/pvldb/BirlerKN24}
\bibfield{author}{\bibinfo{person}{Altan Birler}, \bibinfo{person}{Alfons Kemper}, {and} \bibinfo{person}{Thomas Neumann}.} \bibinfo{year}{2024}\natexlab{}.
\newblock \showarticletitle{Robust Join Processing with Diamond Hardened Joins}.
\newblock \bibinfo{journal}{\emph{Proc. {VLDB} Endow.}} \bibinfo{volume}{17}, \bibinfo{number}{11} (\bibinfo{year}{2024}), \bibinfo{pages}{3215--3228}.
\newblock
\href{https://doi.org/10.14778/3681954.3681995}{doi:\nolinkurl{10.14778/3681954.3681995}}


\bibitem[Dechter(1990)]%
        {DBLP:journals/ai/Dechter90}
\bibfield{author}{\bibinfo{person}{Rina Dechter}.} \bibinfo{year}{1990}\natexlab{}.
\newblock \showarticletitle{Enhancement Schemes for Constraint Processing: Backjumping, Learning, and Cutset Decomposition}.
\newblock \bibinfo{journal}{\emph{Artif. Intell.}} \bibinfo{volume}{41}, \bibinfo{number}{3} (\bibinfo{year}{1990}), \bibinfo{pages}{273--312}.
\newblock
\href{https://doi.org/10.1016/0004-3702(90)90046-3}{doi:\nolinkurl{10.1016/0004-3702(90)90046-3}}


\bibitem[Dechter(2003)]%
        {Dechter2003}
\bibfield{author}{\bibinfo{person}{Rina Dechter}.} \bibinfo{year}{2003}\natexlab{}.
\newblock \bibinfo{booktitle}{\emph{{Constraint Processing}}}.
\newblock \bibinfo{publisher}{Morgan Kaufmann}, \bibinfo{address}{USA}.
\newblock
\showISBNx{1-55860-890-7}


\bibitem[Gottlob et~al\mbox{.}(2016)]%
        {Gottlob2016}
\bibfield{author}{\bibinfo{person}{Georg Gottlob}, \bibinfo{person}{Gianluigi Greco}, \bibinfo{person}{Nicola Leone}, {and} \bibinfo{person}{Francesco Scarcello}.} \bibinfo{year}{2016}\natexlab{}.
\newblock \showarticletitle{{Hypertree Decompositions: Questions and Answers}}. In \bibinfo{booktitle}{\emph{Proceedings of the 35th ACM SIGMOD-SIGACT-SIGAI Symposium on Principles of Database Systems}} (San Francisco, California, USA) \emph{(\bibinfo{series}{PODS '16})}. \bibinfo{publisher}{Association for Computing Machinery}, \bibinfo{address}{New York, NY, USA}, \bibinfo{pages}{57–74}.
\newblock
\showISBNx{9781450341912}
\href{https://doi.org/10.1145/2902251.2902309}{doi:\nolinkurl{10.1145/2902251.2902309}}


\bibitem[Graham(1980)]%
        {graham1980universal}
\bibfield{author}{\bibinfo{person}{M. Graham}.} \bibinfo{year}{1980}\natexlab{}.
\newblock \bibinfo{booktitle}{\emph{On the universal relation}}.
\newblock \bibinfo{type}{Technical Report}. \bibinfo{institution}{University of Toronto, Computer Systems Research Group}.
\newblock


\bibitem[Jr. and Miranker(1994)]%
        {DBLP:journals/ai/BayardoM94}
\bibfield{author}{\bibinfo{person}{Roberto J.~Bayardo Jr.} {and} \bibinfo{person}{Daniel~P. Miranker}.} \bibinfo{year}{1994}\natexlab{}.
\newblock \showarticletitle{An Optimal Backtrack Algorithm for Tree-Structured Constraint Satisfaction problems}.
\newblock \bibinfo{journal}{\emph{Artif. Intell.}} \bibinfo{volume}{71}, \bibinfo{number}{1} (\bibinfo{year}{1994}), \bibinfo{pages}{159--181}.
\newblock
\href{https://doi.org/10.1016/0004-3702(94)90064-7}{doi:\nolinkurl{10.1016/0004-3702(94)90064-7}}


\bibitem[Kolaitis and Vardi(2000)]%
        {DBLP:journals/jcss/KolaitisV00}
\bibfield{author}{\bibinfo{person}{Phokion~G. Kolaitis} {and} \bibinfo{person}{Moshe~Y. Vardi}.} \bibinfo{year}{2000}\natexlab{}.
\newblock \showarticletitle{Conjunctive-Query Containment and Constraint Satisfaction}.
\newblock \bibinfo{journal}{\emph{J. Comput. Syst. Sci.}} \bibinfo{volume}{61}, \bibinfo{number}{2} (\bibinfo{year}{2000}), \bibinfo{pages}{302--332}.
\newblock
\href{https://doi.org/10.1006/JCSS.2000.1713}{doi:\nolinkurl{10.1006/JCSS.2000.1713}}


\bibitem[Leis et~al\mbox{.}(2015)]%
        {Leis2015}
\bibfield{author}{\bibinfo{person}{Viktor Leis}, \bibinfo{person}{Andrey Gubichev}, \bibinfo{person}{Atanas Mirchev}, \bibinfo{person}{Peter Boncz}, \bibinfo{person}{Alfons Kemper}, {and} \bibinfo{person}{Thomas Neumann}.} \bibinfo{year}{2015}\natexlab{}.
\newblock \showarticletitle{{How Good Are Query Optimizers, Really?}}
\newblock \bibinfo{journal}{\emph{Proc. VLDB Endow.}} \bibinfo{volume}{9}, \bibinfo{number}{3} (\bibinfo{date}{Nov.} \bibinfo{year}{2015}), \bibinfo{pages}{204–215}.
\newblock
\showISSN{2150-8097}
\href{https://doi.org/10.14778/2850583.2850594}{doi:\nolinkurl{10.14778/2850583.2850594}}


\bibitem[Miranker et~al\mbox{.}(1997)]%
        {Miranker1997QueryEA}
\bibfield{author}{\bibinfo{person}{Daniel~P. Miranker}, \bibinfo{person}{Roberto~J. Bayardo}, {and} \bibinfo{person}{Vasilis Samoladas}.} \bibinfo{year}{1997}\natexlab{}.
\newblock \showarticletitle{Query Evaluation as Constraint Search; An Overview of Early Results}. In \bibinfo{booktitle}{\emph{International Symposium on the Applications of Constraint Databases}}.
\newblock
\urldef\tempurl%
\url{https://api.semanticscholar.org/CorpusID:8644835}
\showURL{%
\tempurl}


\bibitem[Neumann(2011)]%
        {DBLP:journals/pvldb/Neumann11}
\bibfield{author}{\bibinfo{person}{Thomas Neumann}.} \bibinfo{year}{2011}\natexlab{}.
\newblock \showarticletitle{Efficiently Compiling Efficient Query Plans for Modern Hardware}.
\newblock \bibinfo{journal}{\emph{Proc. {VLDB} Endow.}} \bibinfo{volume}{4}, \bibinfo{number}{9} (\bibinfo{year}{2011}), \bibinfo{pages}{539--550}.
\newblock
\href{https://doi.org/10.14778/2002938.2002940}{doi:\nolinkurl{10.14778/2002938.2002940}}


\bibitem[Ngo et~al\mbox{.}(2018)]%
        {ngo2018wcoja}
\bibfield{author}{\bibinfo{person}{Hung~Q. Ngo}, \bibinfo{person}{Ely Porat}, \bibinfo{person}{Christopher R\'{e}}, {and} \bibinfo{person}{Atri Rudra}.} \bibinfo{year}{2018}\natexlab{}.
\newblock \showarticletitle{{W}orst-{C}ase {O}ptimal {J}oin {A}lgorithms}.
\newblock \bibinfo{journal}{\emph{J. ACM}} \bibinfo{volume}{65}, \bibinfo{number}{3}, Article \bibinfo{articleno}{16} (\bibinfo{date}{March} \bibinfo{year}{2018}), \bibinfo{numpages}{40}~pages.
\newblock
\showISSN{0004-5411}
\href{https://doi.org/10.1145/3180143}{doi:\nolinkurl{10.1145/3180143}}


\bibitem[O'Neil et~al\mbox{.}(2009)]%
        {ONeil2009a}
\bibfield{author}{\bibinfo{person}{Pat O'Neil}, \bibinfo{person}{Betty O'Neil}, {and} \bibinfo{person}{Xuedong Chen}.} \bibinfo{year}{2009}\natexlab{}.
\newblock \bibinfo{booktitle}{\emph{{Star Schema Bechmark - Revision 3, June 5, 2009}}}.
\newblock \bibinfo{type}{resreport}. \bibinfo{institution}{UMass/Boston}.
\newblock
\urldef\tempurl%
\url{https://www.cs.umb.edu/~poneil/StarSchemaB.PDF}
\showURL{%
\tempurl}


\bibitem[{SQLite Documentation}(2024)]%
        {sqliteoptoverview}
\bibfield{author}{\bibinfo{person}{{SQLite Documentation}}.} \bibinfo{year}{2024}\natexlab{}.
\newblock \bibinfo{title}{Query Planning and Optimization}.
\newblock \bibinfo{howpublished}{\url{https://www.sqlite.org/optoverview.html\#hash_joins}}.
\newblock
\newblock
\shownote{Accessed: 2024-07-24}.


\bibitem[(TPC)({[n.\,d.]})]%
        {TPC}
\bibfield{author}{\bibinfo{person}{Transaction Processing Performance~Council (TPC)}.} \bibinfo{year}{[n.\,d.]}\natexlab{}.
\newblock \bibinfo{title}{{TPC-H Benchmark}}.
\newblock \bibinfo{howpublished}{Online}.
\newblock
\urldef\tempurl%
\url{http://tpc.org/tpc_documents_current_versions/pdf/tpc-h_v3.0.0.pdf}
\showURL{%
\tempurl}
\newblock
\shownote{Accessed on 11-18-2021}.


\bibitem[Tziavelis et~al\mbox{.}(2022)]%
        {DBLP:journals/corr/abs-2205-05649}
\bibfield{author}{\bibinfo{person}{Nikolaos Tziavelis}, \bibinfo{person}{Wolfgang Gatterbauer}, {and} \bibinfo{person}{Mirek Riedewald}.} \bibinfo{year}{2022}\natexlab{}.
\newblock \showarticletitle{Any-k Algorithms for Enumerating Ranked Answers to Conjunctive Queries}.
\newblock \bibinfo{journal}{\emph{CoRR}}  \bibinfo{volume}{abs/2205.05649} (\bibinfo{year}{2022}).
\newblock
\href{https://doi.org/10.48550/ARXIV.2205.05649}{doi:\nolinkurl{10.48550/ARXIV.2205.05649}}
\showeprint[arXiv]{2205.05649}


\bibitem[Tziavelis et~al\mbox{.}(2024)]%
        {DBLP:journals/sigmod/TziavelisGR24}
\bibfield{author}{\bibinfo{person}{Nikolaos Tziavelis}, \bibinfo{person}{Wolfgang Gatterbauer}, {and} \bibinfo{person}{Mirek Riedewald}.} \bibinfo{year}{2024}\natexlab{}.
\newblock \showarticletitle{Ranked Enumeration for Database Queries}.
\newblock \bibinfo{journal}{\emph{{SIGMOD} Rec.}} \bibinfo{volume}{53}, \bibinfo{number}{3} (\bibinfo{year}{2024}), \bibinfo{pages}{6--19}.
\newblock
\href{https://doi.org/10.1145/3703922.3703924}{doi:\nolinkurl{10.1145/3703922.3703924}}


\bibitem[Veldhuizen(2012)]%
        {DBLP:journals/corr/abs-1210-0481}
\bibfield{author}{\bibinfo{person}{Todd~L. Veldhuizen}.} \bibinfo{year}{2012}\natexlab{}.
\newblock \showarticletitle{Leapfrog Triejoin: a worst-case optimal join algorithm}.
\newblock \bibinfo{journal}{\emph{CoRR}}  \bibinfo{volume}{abs/1210.0481} (\bibinfo{year}{2012}).
\newblock
\showeprint[arXiv]{1210.0481}
\urldef\tempurl%
\url{http://arxiv.org/abs/1210.0481}
\showURL{%
\tempurl}


\bibitem[Yannakakis(1981)]%
        {DBLP:conf/vldb/Yannakakis81}
\bibfield{author}{\bibinfo{person}{Mihalis Yannakakis}.} \bibinfo{year}{1981}\natexlab{}.
\newblock \showarticletitle{Algorithms for Acyclic Database Schemes}. In \bibinfo{booktitle}{\emph{Very Large Data Bases, 7th International Conference, September 9-11, 1981, Cannes, France, Proceedings}}. \bibinfo{publisher}{{IEEE} Computer Society}, \bibinfo{pages}{82--94}.
\newblock


\bibitem[Yu and Ozsoyoglu(1979)]%
        {Yu1979AnAF}
\bibfield{author}{\bibinfo{person}{Clement~T. Yu} {and} \bibinfo{person}{M.~Z. Ozsoyoglu}.} \bibinfo{year}{1979}\natexlab{}.
\newblock \showarticletitle{An algorithm for tree-query membership of a distributed query}. In \bibinfo{booktitle}{\emph{Annual International Computer Software and Applications Conference}}.
\newblock
\urldef\tempurl%
\url{https://api.semanticscholar.org/CorpusID:7812638}
\showURL{%
\tempurl}


\bibitem[Zhu et~al\mbox{.}(2017)]%
        {DBLP:journals/pvldb/ZhuPSP17}
\bibfield{author}{\bibinfo{person}{Jianqiao Zhu}, \bibinfo{person}{Navneet Potti}, \bibinfo{person}{Saket Saurabh}, {and} \bibinfo{person}{Jignesh~M. Patel}.} \bibinfo{year}{2017}\natexlab{}.
\newblock \showarticletitle{Looking Ahead Makes Query Plans Robust}.
\newblock \bibinfo{journal}{\emph{Proc. {VLDB} Endow.}} \bibinfo{volume}{10}, \bibinfo{number}{8} (\bibinfo{year}{2017}), \bibinfo{pages}{889--900}.
\newblock
\href{https://doi.org/10.14778/3090163.3090167}{doi:\nolinkurl{10.14778/3090163.3090167}}


\end{thebibliography}

\nop{\appendix
\label{sec:appendix}

\section{Proof of Theorem~\ref{thm:ttj-cyclic}}\label{sec:full-thm-ttj-cyclic}

\rw{The proof is actually not very complicated, so please move it to the main text.}

Given a rooted tree convolution $\tconv(Q)$, we generate a query plan $p$ as follows.
Suppose $\tconv(Q)$ consists of $m$ nested trees $T_1, \dots, T_m$
where $T_1$ is the most deeply nested tree and $T_m$ is the outermost tree.
Let $p_1, \dots, p_m$ be the plans corresponding to $T_1, \dots, T_m$.
Then, each $p_i$ is a plan that corresponds to the reverse of a GYO-reduction order of $T_i$,
where the first relation in $p_i$ is result of $p_{i-1}$ for $i>1$.
The outermost plan $p_m$ then computes the final result of $Q$.

\rw{Change $C$ to $\tconv$.}

\nop{Any $p_i$ for $2 \le i \le m$ contains $p_{i-1}$ with extra atoms in $Q$ that are in $p_i$ but not in $p_{i-1}$.}

\nop{\begin{proposition}{\label{prop:cyclic-backjumping}}
\nop{During $\ttj$ execution on a given rooted convolution $\tconv$ of a query $Q$, $\ttj$ either backjumps to an atom in $Q$ and deletes a tuple from the atom or backjumps to the last relation of the previous nesting level.}
During $\ttj$ execution on a given rooted convolution $\tconv$ of a query $Q$, if a lookup fails at $R$ that belongs to $p_i$, $\ttj$ either backjumps to an atom in $p_i$ and deletes a tuple from the atom (except the root relation) or backjumps to the last atom of $p_{i-1}$ if $i \ge 2$.
\nop{closest relation to $R$ in $p_{i-1}$ if $i \ge 2$.}
\end{proposition}

\begin{proof}
Since $p_1$ is the reverse of a GYO-reduction order of the most deeply nested tree in $\tconv$, by \cref{lem:gyo}, for any $R$ except the root relation, \lstinline|ttj| backjumps to $R$'s parent. If $R$'s parent is non-root relation, a tuple is deleted. For any $p_i$ with $i \ge 2$, since $p_i$ is the reverse of a GYO-reduction order of the $i$-th tree of $\tconv$ and by the modified \lstinline|parent| function, every relation in $p_i$ but not $p_{i-1}$ has parent. 
\end{proof}}

\begin{proposition}{\label{prop:cyclic-backjumping}}
During $\ttj$ execution on a given rooted convolution $\tconv$ of a query $Q$, if a lookup fails at $R$ that belongs to $p_i$ but not in $p_{i-1}$ ($i \ge 2$), $\ttj$ either backjumps to an atom that is in $p_i$ but not in $p_{i-1}$ or backjumps to the last atom of $p_{i-1}$.
\end{proposition}

\begin{proof}
For any $p_i$ with $i \ge 2$, since $p_i$ is the reverse of a GYO-reduction order of the $i$-th tree of $\tconv$ and by the modified \lstinline|parent| function, every relation in $p_i$ but not in $p_{i-1}$ has parent. Furthermore, if lookup fails at $R$ that is in $p_i$ but not in $p_{i-1}$, \lstinline|ttj| backjumps to $R$'s parent. Then, the result follows by the definition of the modified \lstinline|parent| function.
\end{proof}

In $p_1$, if the parent of a relation is the first relation of $p_1$, \lstinline|ttj| backjumps to the first relation. We can treat the first relation of $p_1$ as $p_0$ (i.e., the most deeply nested tree in $\tconv$ is now a node of an atom of $Q$) and the relation is also the last atom of $p_0$. Therefore, we can remove the restriction of $i \ge 2$ in \cref{prop:cyclic-backjumping}. In the following proof, we reference the \cref{prop:cyclic-backjumping} with the understanding that it holds for $i \ge 1$.

\ttjcyclic*
\begin{proof}
Since there is no change to \lstinline|ttj| except using the modified
\lstinline|parent| function for $\tconv$, like Proof of \cref{thm:linear}, we only
need to show the algorithm makes $O(|\IN| + |\OUT| + \sum_i |S_i|)$ number of calls to
\lstinline|ttj|.

Excep for the initial call to \lstinline|ttj| with \lstinline|ttj((), p, 1)|, every call to \lstinline|ttj| has three possible outcomes: 
\begin{enumerate*}
	\item It outputs a tuple.\label{enum:output-a}
	\item It backjumps and possibly deletes a tuple from an input relation.\label{enum:backjump-a}
	\item It recursively calls \lstinline|ttj|.\label{enum:recursive-a}
\end{enumerate*}
Because the query plan has constant length, 
there can be at most a constant number of recursive calls
to \lstinline|ttj| (case~\ref{enum:recursive-a}) until we reach cases~\ref{enum:output-a} or~\ref{enum:backjump-a}.
There can be $O(|Q|)$ \lstinline|ttj| calls for case~\ref{enum:output-a}. 
Since lookup cannot fail at the first relation of $p$, the relation that lookup fails at is in some $p_i$ but not in $p_{i-1}$. Let $R$ be a relation that a lookup fails at. By \cref{prop:cyclic-backjumping}, there can be two cases on where \lstinline|ttj| backjumps to. If \lstinline|ttj| backjumps to an atom that is also in $p_i$ but not in $p_{i-1}$, a tuple is deleted. This case can happen $O(|\IN|)$ times. If \lstinline|ttj| backjumps to the last atom of $p_{i-1}$, since \lstinline|ttj| works no different than binary join from this moment until next lookup failure, this can happen $O(|S_{i-1}|)$ times; the tuples computed at the last atom of $p_{i-1}$ is $S_{i-1}$.
Therefore, given $1 \le i \le m$, there are at most $O(|\IN| + |\OUT| + \sum_{i=1}^{m-1} |S_i|)$ calls to \lstinline|ttj|, and the algorithm runs in that time.
\end{proof}

\nop{
We prove the number of calls to \lstinline|ttj| by induction on $i$. Base case $i = 1$.}}

\end{document}